\documentclass[onecolumn,12pt]{IEEEtran}
\usepackage{graphicx,epic,eepic,epsfig,amsmath,amsfonts,latexsym,amssymb,color,enumerate,bbm}

\newtheorem{theorem}{Theorem}
\newtheorem{definition}{Definition}
\newtheorem{lemma}{Lemma}
\newtheorem{remark}{Remark}

\newtheorem{proposition}{Proposition}

\def\QED{\mbox{\rule[0pt]{1.5ex}{1.5ex}}}

\newcommand{\qed}{\hfill \QED}
\newcommand{\tr}{\operatorname{Tr}}
\newcommand{\mc}{\mathcal}
\newcommand{\ket}[1]{|#1\rangle}
\newcommand{\bra}[1]{\langle#1|}
\newcommand{\proj}[1]{| #1\rangle\!\langle #1 |}
\newcommand{\beq}{\begin{equation}}
\newcommand{\eeq}{\end{equation}}
\newcommand{\rar}{\rightarrow}

\newcommand{\1}{\mathbbm{1}}
\newcommand{\supp}{\operatorname{supp}}
\newcommand{\ox}{\otimes}
\newenvironment{proofof}[1]{\vspace*{5mm} \par \noindent
{\it Proof of #1:\hspace{2mm}}}{\qed
}

\begin{document}
\title{Tight Exponential Analysis for Smoothing the Max-Relative Entropy and for Quantum Privacy Amplification}

\author{Ke~Li, Yongsheng~Yao, and Masahito~Hayashi~\IEEEmembership{Fellow,~IEEE}
\thanks{Ke Li is with the Institute for Advanced Study in Mathematics, Harbin Institute of
Technology, Nangang District,  Harbin 150001, China. (carl.ke.lee@gmail.com). Yongsheng Yao is with the Institute for Advanced Study in Mathematics, School of Mathematics, Harbin Institute of
Technology, Nangang District,  Harbin 150001, China, and Shenzhen Institute for Quantum Science and Engineering, Southern University of Science and Technology, Nanshan District,
Shenzhen, 518055, China. (yongsh.yao@gmail.com).
Masahito Hayashi is with
Shenzhen Institute for Quantum Science and Engineering, Southern University of Science and Technology, Nanshan District,
Shenzhen, 518055, China,
International Quantum Academy (SIQA), Futian District, Shenzhen 518048, China,
Guangdong Provincial Key Laboratory of Quantum Science and Engineering,
Southern University of Science and Technology, Nanshan District, Shenzhen 518055, China,
and
Graduate School of Mathematics, Nagoya University, Nagoya, 464-8602, Japan.
(e-mail:hayashi@sustech.edu.cn, masahito@math.nagoya-u.ac.jp)}}
\date{}

\maketitle
\begin{abstract}
The max-relative entropy together with its smoothed version is a basic
tool in quantum information theory. In this paper, we derive the exact
exponent for the asymptotic decay of the small modification of the quantum
state in smoothing the max-relative entropy based on purified distance.
We then apply this result to
the problem of privacy amplification against quantum side information, and
we obtain an upper bound for the exponent of the asymptotic decreasing of
the insecurity, measured using either purified distance or relative entropy.
Our upper bound complements the earlier lower bound established by
Hayashi, and the two bounds match when the rate of randomness extraction
is above a critical value. Thus, for the case of high rate, we have
determined the exact security exponent. Following this, we give examples
and show that in the low-rate case, neither the upper bound nor the lower
bound is tight in general. This exhibits a picture similar to that of
the error exponent in channel coding.
Lastly, we investigate the asymptotics of equivocation and its exponent under
the security measure using the sandwiched R{\'e}nyi divergence of order
$s\in (1,2]$, which has not been addressed previously in the quantum setting.
\end{abstract}


\section{Introduction}
\label{sec:Introduction}
The smooth max-relative entropy is a basic tool in quantum information
theory~\cite{Renner2005security, Datta2009min, TCR2009fully,
TCR2010duality, BrandaoPlenio2010reversible, BCR2011the, Tomamichel2015quantum},
developed in parallel with the related but different concepts of hypothesis
testing relative entropy~\cite{Han2000, NagaokaHayashi, HayashiNagaoka,
WangRenner2012one, TomamichelHayashi2013hierarchy, DKFRR2014generalized,
MatthewsWehner2014finite} and information spectrum relative
entropy~\cite{HanVerdu, Hanbook, Han2000, NagaokaHayashi, HayashiNagaoka,
TomamichelHayashi2013hierarchy}. In the asymptotic limit when multiple
copies of underlying resource states are available, the one-shot
characterizations using these quantities lead to
results of the traditional information-theoretic type. Indeed, the quantum
relative entropy (Kullback-Leibler divergence), arguably, finds its most direct
operational interpretations in the asymptotic analysis of these
quantities~\cite{HiaiPetz1991proper, OgawaNagaoka2000strong,
TCR2009fully, Tomamichel2015quantum, HanVerdu, Hanbook, NagaokaHayashi}.
The asymptotic analysis of the quantum hypothesis testing entropy has been
extended to the second-order regime~\cite{TomamichelHayashi2013hierarchy,
Li2014second}. Moreover, large-deviation type exponential analysis for
quantum hypothesis testing is well understood~\cite{NussbaumSzkola2009chernoff,
AKCMBMA2007discriminating, Nagaoka2006converse, Hayashi2007error,
ANSV2008asymptotic, MosonyiOgawa2015quantum}, and the relation between
information spectrum and quantum hypothesis testing has been well
studied up to the exponential analysis~\cite{Han2000, HayashiNagaoka}.

However, the details of the asymptotic behavior of the smoothing of
the max-relative entropy is more complicated. In particular, it depends
on the choice of the distance measure to define the smoothing. Originally,
Renner~\cite{Renner2005security} defined the smoothing of the max-relative
entropy based on the trace norm distance. The paper~\cite{Hayashi2016} derived
its type exponential behavior based on the trace norm distance in the
classical case. Later, the references~\cite{TCR2009fully, TCR2010duality,
TSSR2011leftover} introduced another smoothing of the max-relative entropy
based on the purified distance. The
reference~\cite{TomamichelHayashi2013hierarchy} showed that this type of
smoothing of the max-relative entropy has the same behavior with the quantum
hypothesis testing entropy in the second-order regime. The exact large-deviation
type exponential behavior in both types of smoothing of the max-relative
entropy remains unclear in the general case.

In this paper, we conduct the exponential analysis for the smoothing of the
max-relative entropy based on the purified distance.
For two quantum states $\rho$ and $\sigma$, consider the
smoothing quantity $\epsilon(\rho^{\ox n} \| \sigma^{\ox n}, nr):=\min\{\ell(
\rho^{\ox n},\widetilde{\rho^n})\ |\ \widetilde{\rho^n} \leq 2^{nr}\sigma^{\ox n}\}$,
where $\ell$ is certain distance measure and $\widetilde{\rho^n}$ is a
subnormalized quantum state. It is known that when $r$ is larger than the
relative entropy $D(\rho\|\sigma)$, the smoothing quantity can be arbitrarily
small when $n$ is big enough. We determine the precise exponent under which
the smoothing quantity converges to $0$ exponentially in this case for $\ell$
being the purified distance (cf. Theorem~\ref{theorem:exp-mre}). Remarkably,
this exponent is given in terms of the sandwiched R{\'e}nyi
divergence~\cite{MDSFT2013on, WWY2014strong}. Our result naturally covers the
exponential analysis for the smoothing of a particular type of the conditional
min-entropy and the max-mutual information (see, e.g., Definition~\ref{def:scme}
for the conditional min-entropy).

We apply the above-mentioned result to the problem of private randomness extraction
against quantum side information, a quantum information processing task also called
privacy amplification~\cite{BBCM1995generalized, DevetakWinter2005distillation,
Renner2005security, TomamichelHayashi2013hierarchy, Hayashi2015precise, Dupuis2021privacy}.
Asymptotic security of privacy amplification in the i.i.d. situation is known to
hold when the rate of randomness extraction does not exceed the conditional entropy
of the raw randomness given the quantum side information~\cite{DevetakWinter2005distillation,Renner2005security}.
We obtain an upper bound for the rate of exponential decreasing of the insecurity in
the i.i.d. situation, measured either by the purified distance or by the relative
entropy, in terms of a version of the sandwiched R{\'e}nyi conditional entropy (cf. Theorem~\ref{theorem:exp-privacy-up}). Notice that an upper bound for the rate of
exponential decreasing corresponds to a lower bound of the insecurity.
This complements the previous work \cite{Hayashi2015precise}, which has established a privacy
amplification theorem concerning the achievability via two-universal hash
functions and obtains a corresponding lower bound for the exponent in the asymptotic case.
We show that our upper bound matches the above mentioned
lower bound when the rate $R$ of randomness extraction is above a critical value $R_\text{critical}$.
Thus, for the case with high rate of randomness extraction, we have determined
the exact security exponent (cf. Theorem~\ref{theorem:exp-privacy}). For the
low-rate situation, we give simple examples to show that neither the upper bound
nor the lower bound is tight in general. These results exhibit a picture similar
to that of the error exponent of channel coding in classical information
theory~\cite{Gallager1968information}.

In addition, we investigate the security of privacy amplification under a more
general class of information measure---the sandwiched R{\'e}nyi divergence of order
$s\in (1,2]$. We prove tight equivocation rate for this security measure and derive
the exponential rate of decay of the insecurity. This problem has been analyzed
by the reference~\cite{HayashiTan2016equivocations} in the classical case, in
which they evaluate the asymptotics of equivocations and their exponents under
various R{\'e}nyi information measures. We generalize their results to the quantum
setting here.

Our results provide operational interpretations to the sandwiched R{\'e}nyi
divergence and the sandwiched R{\'e}nyi conditional entropy, in addition to
previous operational interpretations to the sandwiched R{\'e}nyi information quantities~\cite{MosonyiOgawa2015quantum, MosonyiOgawa2015two, CMW2016strong,
HayashiTomamichel2016correlation, MosonyiOgawa2017strong, CHDH2020non}.
However, the operational interpretations found in the present paper, as well
as those in the concurrent work of~\cite{LiYao2021reliability} which addresses
different problems, are in stark
contrast to those of the previous ones, in the sense explained as follows. The
works~\cite{MosonyiOgawa2015quantum, MosonyiOgawa2015two,CMW2016strong,
HayashiTomamichel2016correlation, MosonyiOgawa2017strong, CHDH2020non} proved
that the sandwiched R{\'e}nyi information quantities characterize the strong
converse exponents, that is, the exponential rates under which the underlying
error goes to $1$. Our results and the work~\cite{LiYao2021reliability}, for the
first time, show that the sandwiched R{\'e}nyi information quantities characterize
the exponents under which the underlying error goes to $0$, and are therefore of
greater realistic significance.

The remainder of this paper is organized as follows. In Section~\ref{sec:preliminaries},
we introduce the necessary notations, definitions and some properties
of quantum entropic quantities. Then in Section~\ref{sec:smooth-entropy}, we derive the
optimal exponent in smoothing the max-relative entropy. Section~\ref{sec:privacy-amp}
is devoted to the analysis of the asymptotic rates of exponential decreasing
of the insecurity of privacy amplification. In Section~\ref{sec:equivocation},
we investigate the equivocation rate and the exponential rate of decay of the insecurity
of privacy amplification measured by the sandwiched R{\'e}nyi divergence of order
$s\in (1,2]$. At last, in
Section~\ref{sec:discussion}, we conclude the paper with some discussion and open
questions.

\section{Notation and preliminaries}
\label{sec:preliminaries}
\subsection{Basic notation}
Let $\mc{H}$ be a finite dimensional Hilbert space. $\mc{L}(\mc{H})$ denotes
the set of linear operators on $\mc{H}$, and $\mc{P}(\mc{H})\subset\mc{L}
(\mc{H})$ denotes the set of positive semidefinite operators. $\1_{\mathcal{H}}$
is the identity operator. The set of (normalized) quantum states and
subnormalized quantum states on $\mc{H}$ are denoted as $\mc{S}(\mc{H})$
and $\mathcal{S}_{\leq}(\mathcal{H})$, respectively. They are given by
\begin{align*}
\mc{S}(\mc{H})        &=\{\rho \in \mc{P}(\mc{H}) | \tr\rho=1 \},\\
\mc{S}_{\leq}(\mc{H}) &=\{\rho \in \mc{P}(\mc{H}) | \tr\rho \leq 1 \}
\end{align*}
and also called density operators. A classical-quantum (CQ) state is a
bipartite state of the form $\rho_{XA}=\sum_x p(x) \proj{x}_X \ox \rho^x_A$,
where $\rho^x_A \in \mc{S}(\mc{H})$, $p(x)$ is a probability distribution,
and $\{\ket{x}\}$ is an orthonormal basis of the underlying Hilbert space
$\mathcal{H}_X$. If the system $X$ is classical as in the CQ state, we also
use the notation $X$ to represent a random variable that takes the value
$x$ with probability $p(x)$. The set of all the possible values of $X$ is
denoted by the corresponding calligraphic letter $\mc{X}$.

We write $A\geq 0$ if $A\in\mc{P}(\mc{H})$, and $A\geq B$ if $A-B\geq 0$.
If $A \in \mathcal{L}(\mathcal{H})$ is self-adjoint, we use $\{A \geq 0\}$
to denote the spectral projection of $A$ corresponding to all non-negative
eigenvalues. $\{A > 0\}$, $\{A \leq 0\}$ and $\{A < 0\}$ are defined in
a similar way. The positive part of $A$ is defined as $A_+:=A\{A>0\}$. We
can easily check that, for any $D \in \mathcal{L}(\mathcal{H})$ such that
$0\leq D\leq\1$,
\begin{equation}
\label{eq:pos}
\tr A_+ \geq \tr AD.
\end{equation}

A quantum channel (or quantum operation), which acts on quantum states,
is formally described by a linear, completely positive, trace-preserving
(CPTP) map $\Phi:\mc{L}(\mc{H}_A) \rightarrow \mc{L}(\mc{H}_B)$. A quantum
measurement is described by a set of positive semidefinite operators
$\mc{M}=\{M_x\}_x$ such that $\sum_xM_x=\1$, and it converts a quantum state
$\rho$ into a probability vector $\vec{p}$ with $\vec{p}_x=\tr\rho M_x$.
For each quantum measurement $\mc{M}=\{M_x\}_x$, there is a measurement
channel $\Phi_\mc{M}:\rho\mapsto \sum_x(\tr\rho M_x)\proj{x}$, where
$\{\ket{x}\}$ is an orthonormal basis.

We employ the purified distance~\cite{GLN2005distance, TCR2009fully}
to measure the closeness of two states $\rho, \sigma \in \mathcal{S}_{\leq}
(\mathcal{H})$. It is defined as $P(\rho,\sigma):=\sqrt{1-F^2(\rho,\sigma)}$,
where
\[
F(\rho,\sigma):=\left\|\sqrt\rho \sqrt\sigma \right\|_1
                              + \sqrt{(1-\tr\rho)(1-\tr \sigma)} \\
\]
is the fidelity function. The purified distance has some nice properties,
inherited from the fidelity.
\begin{proposition}
\label{proposition:pd-properties}
The following properties hold for the purified distance.
\begin{enumerate}[(i)]
  \item Triangle inequality~\cite{GLN2005distance}: Let $\rho,\sigma,\tau
        \in \mc{S}_{\leq}(\mc{H})$. Then
        \[
        P(\rho,\sigma) \leq P(\rho,\tau) + P(\tau,\sigma) ;
        \]
  \item Fuchs-van de Graaf inequality~\cite{FuchsVan1999cryptographic}:
        Let $\rho,\sigma \in \mc{S}_{\leq}(\mc{H})$. Then
        \[
        d(\rho,\sigma) \leq P(\rho,\sigma)
                       \leq \sqrt{2d(\rho,\sigma)-d^2(\rho,\sigma)},
        \]
        where $d(\rho,\sigma):=\frac{1}{2}(\left\| \rho-\sigma \right\|_1
        + |\tr (\rho-\sigma)|)$ is the trace distance;
  \item Data processing inequality~\cite{BCFJS1996noncommuting}: Let $\rho,
        \sigma \in \mc{S}_{\leq}(\mc{H})$ and $\Phi$ be a CPTP map. Then
        \[
        P(\rho,\sigma) \geq P(\Phi(\rho),\Phi(\sigma)) ;
        \]
  \item Uhlmann's theorem~\cite{Uhlmann1976the}: Let $\rho_{AB}
        \in \mc{S}_{\leq}(\mc{H_{AB}})$ be a bipartite state, and
        $\sigma_{A} \in \mc{S}_{\leq}(\mc{H_{A}})$. Then there exists
        an extension $\sigma_{AB}$ of $\sigma_{A}$ such that
        \[
        P(\rho_{AB},\sigma_{AB}) = P(\rho_{A},\sigma_{A}).
        \]
\end{enumerate}
\end{proposition}
The $\epsilon\text{-ball}$ of subnormalized quantum states around $\rho\in
\mc{S}(\mc{H})$ is defined using the purified distance as
\[
  \mathcal{B}^{\epsilon}(\rho):=\{\tilde{\rho} \in \mc{S}_{\leq}(\mc{H}) |
  P(\tilde{\rho}, \rho) \leq \epsilon\}.
\]

For an operator $A\in\mc{L}(\mc{H})$, let $v(A)$ be the number of
different eigenvalues of $A$. If $A$ is self-adjoint with spectral
projections $P_1, \ldots, P_{v(A)}$, then the associated pinching
map $\mc{E}_{A}: \mc{L}(\mc{H}) \rightarrow \mc{L}(\mc{H})$ is a
CPTP map given by
\[
\mc{E}_A : X \mapsto \sum_{i} P_i X P_i.
\]
The pinching inequality~\cite{Hayashi2002optimal} states that if
$X$ is positive semidefinite, we have
\begin{equation}
\label{eq:pinchingineq}
X \leq v(A) \mc{E}_A(X).
\end{equation}

\subsection{Entropies and information divergences}
The quantum relative entropy for $\rho\in \mc{S}(\mc{H})$ and
 $\sigma\in \mc{P}(\mc{H})$ is defined~\cite{Umegaki1954conditional} as
\[
  D(\rho\|\sigma) :=
  \begin{cases}
  \tr(\rho(\log\rho-\log\sigma)) & \text{ if }\supp(\rho)\subseteq\supp(\sigma), \\
  +\infty                        & \text{ otherwise},
  \end{cases}
\]
where the logarithm function $\log$ is with base $2$ throughout this paper.
For a bipartite state $\rho_{AB} \in \mc{S}(\mc{H}_{AB})$, the quantum
mutual information and the conditional entropy are defined,
respectively, as
\[
\begin{split}
I(A:B)_\rho &:=D(\rho_{AB} \| \rho_A \ox \rho_B), \\
H(A|B)_\rho &:=-D(\rho_{AB} \| \1_A \ox \rho_B).
\end{split}
\]

Among various inequivalent generalizations of the R{\'e}nyi relative entropy
to the non-commutative quantum situation, the sandwiched R{\'e}nyi
divergence~\cite{MDSFT2013on, WWY2014strong} is of particular interest.
\begin{definition}
\label{definition:SRD}
Let $\rho \in \mc{S}(\mc{H})$, $\sigma \in \mc{P}(\mc{H})$, and $\alpha\in (0,1)
\cup (1,\infty)$. If either $0<\alpha<1$ and $\tr\rho\sigma\neq 0$ or $\alpha>1$
and $\supp(\rho)\subseteq\supp(\sigma)$, the sandwiched R{\'e}nyi divergence of
order $\alpha $ is defined as
\[
D_{\alpha}(\rho \| \sigma):=\frac{1}{\alpha-1} \log Q_{\alpha}(\rho \| \sigma),
\quad \text{where }\
Q_{\alpha}(\rho \| \sigma):=\tr {({\sigma}^{\frac{1-\alpha}{2\alpha}} \rho
                                  {\sigma}^{\frac{1-\alpha}{2\alpha}})}^\alpha.
\]
Otherwise, we set $D_{\alpha}(\rho \| \sigma)=+\infty$.
\end{definition}

When $\alpha$ goes to infinity, $D_{\alpha}(\rho \| \sigma)$ converges
to the max-relative entropy~\cite{Datta2009min}
\begin{equation}
\label{equ:maxdefinition}
D_\text{max}(\rho \| \sigma):=\inf\{ \lambda \;|\; \rho \leq 2^{\lambda}\sigma\}.
\end{equation}
For $\rho_{AB} \in \mc{S}(\mc{H}_{AB})$ and $\alpha\in (0,1)\cup (1,\infty)$,
we consider the sandwiched R{\'e}nyi conditional entropy of order
$\alpha$ defined as~\cite{TBH2014relating}
\[
H_{\alpha}(A|B)_\rho:=-D_{\alpha}(\rho_{AB} \| \1_A \ox \rho_B).
\]
If the system $B$ is of dimension $1$, the sandwiched R{\'e}nyi conditional
entropy reduces to the R{\'e}nyi entropy of a single system $H_{\alpha}(A)_\rho
=-D_{\alpha}(\rho_{A} \| \1_A )=\frac{1}{1-\alpha}\log\tr\rho_A^\alpha$. We
mention that these definitions can be extended to include the cases that
$\alpha=0,1,+\infty$ by taking the limit of $\alpha$. Moreover, for a CQ
state $\rho_{XE}$, we define for $s>0$,
\beq\label{eq:Rs}
\hat{R}(s):=\frac{\mathrm{d}}{\mathrm{d} s} sH_{1+s}(X|E)_{\rho}
\eeq
and we set
\beq\label{eq:Rcritical}
  R_{\rm critical}:=\hat{R}(1)
 =\frac{\mathrm{d}}{\mathrm{d} s} sH_{1+s}(X|E)_{\rho}\big|_{s=1}.
\eeq

In the next proposition,
we collect a few properties of the R{\'e}nyi quantities defined above.
\begin{proposition}
\label{prop:srd-properties}
Let $\rho \in \mc{S}(\mc{H}), \sigma \in \mc{P}(\mc{H})$, $\xi_{AB} \in \mc{S}
(\mc{H}_{AB})$, and $\omega_{XAB}=\sum_x p(x) \proj{x}_X \ox \omega^x_{AB} \in
\mc{S}(\mc{H}_{XAB})$. Then the sandwiched R{\'e}nyi divergence and the
R{\'e}nyi conditional entropy satisfy the following properties:
\begin{enumerate}[(i)]
  \item Monotonicity~\cite{MDSFT2013on,Beigi2013sandwiched}: If $0<\alpha\leq
        \beta$, then $D_{\alpha}(\rho\|\sigma)\leq D_{\beta}(\rho \| \sigma)$ ;
  \item Limit of $\alpha \rightarrow 1$~\cite{MDSFT2013on,WWY2014strong}:
        $\lim\limits_{\alpha\rightarrow 1}D_{\alpha}(\rho\| \sigma)=D(\rho\|
        \sigma)$, and $\lim\limits_{\alpha\rightarrow 1}H_{\alpha}(A|B)_\rho=
        H(A|B)_\rho$ ;
  \item Data processing inequality~\cite{FrankLieb2013monotonicity,
        Beigi2013sandwiched, MDSFT2013on, WWY2014strong}: Let $\alpha \in
        [\frac{1}{2}, \infty)$ and $\Phi$ be a CPTP map. Then
        \[
          D_{\alpha}(\rho\| \sigma)\geq D_{\alpha}(\Phi(\rho)\| \Phi(\sigma)) ;
        \]
  \item Convexity~\cite{MosonyiOgawa2017strong}: For $\alpha \in (0, +\infty)$,
        the function $f(\alpha)=\log Q_\alpha(\rho \| \sigma)$ is convex;
  \item Invariance under isometries~\cite{MDSFT2013on,WWY2014strong}: Let
        $\mc{U}:\mc{H} \rightarrow \mc{H}'$, $\mc{U}_A: \mc{H}_A \rightarrow
        \mc{H}_{A'}$ and $\mc{U}_B: \mc{H}_B \rightarrow \mc{H}_{B'}$ be
        isometries. Then $D_{\alpha}(\mc{U}\rho\mc{U}^{*} \| \mc{U} \sigma
        \mc{U}^{*})=D_{\alpha}(\rho \| \sigma)$ and $H_\alpha(A' | B')_
        {(\mc{U}_A\ox \mc{U}_B)\xi_{AB}(\mc{U}^{*}_A\ox \mc{U}^{*}_B)}=H_\alpha
        (A|B)_{\xi_{AB}}$ ;
  \item Monotonicity under discarding classical information~\cite{LWD2016strong}:
        For the state $\sigma_{XAB}$ that is classical on $X$ and for $\alpha\in
        (0,+\infty)$,
        \[H_{\alpha}(AX|B)_\sigma \geq H_{\alpha}(A|B)_\sigma\, .\]
\end{enumerate}
\end{proposition}

The result of the following proposition is established by
Mosonyi and Ogawa~\cite{MosonyiOgawa2015quantum}.
\begin{proposition}
\label{prop:MosonyiOgawa}
For any $\rho \in \mc{S}(\mc{H})$, $\sigma \in \mc{P}(\mc{H})$,
$a \in \mathbb{R}$ and $t>0$, we have
\begin{equation}
\label{eq:MosonyOgawa}
  \lim_{n\rightarrow\infty}
  \frac{\log \tr \rho^{\ox n} \{\rho^{\ox n}>t2^{na}\sigma^{\ox n}\}}{n}
=\lim_{n\rar\infty}\frac{\log\tr(\rho^{\ox n}-t2^{na}\sigma^{\ox n})_+}{n}
=\inf_{s \geq 0}\big\{ s\big(D_{1+s}(\rho \| \sigma)-a\big)\big\}.
\end{equation}
\end{proposition}

\begin{remark}
\label{remark:sim}
In its original statement~\cite{MosonyiOgawa2015quantum}, Proposition~\ref{prop:MosonyiOgawa}
appears with $t$ being $1$ and $a$ being in the interval $(D(\rho \| \sigma), D_{\rm{max}}(\rho \| \sigma))$. However, it is easy to see that it holds for any $t>0$ and $a\in\mathbb{R}$. The reason for $t$ is obvious, since it can be absorbed into $a$ when $n\rar\infty$. As for $a$, we discuss the following two cases. 1) $a > D_{\rm{max}}(\rho \|\sigma)$: it is easy to check that the three expressions in Eq.~(\ref{eq:MosonyOgawa}) are all $-\infty$. 2) $a < D(\rho \|\sigma)$: by the equalities of Eq.~(\ref{eq:MosonyOgawa}) established for $a \in (D(\rho \| \sigma),D_{\rm{max}}(\rho \| \sigma))$, we see that the two limits in Eq.~(\ref{eq:MosonyOgawa}) goes to $0$ when $a\searrow D(\rho \| \sigma)$. In addition, $\tr \rho^{\ox n} \{\rho^{\ox n}>2^{na}\sigma^{\ox n}\}$ and $\tr(\rho^{\ox n}-2^{na}\sigma^{\ox n})_+$ are monotonically decreasing with $a$ (cf.~\cite{NagaokaHayashi}). So, the two limits are nonnegative when $a < D(\rho \| \sigma)$. On the other hand, it is easy to see that the two limits are upper bounded by $0$ because the terms in the logarithm function are upper bounded by $1$. Hence, we conclude that the two limits actually equal to $0$ when $a < D(\rho \| \sigma)$. This coincides with the third expression of Eq.~(\ref{eq:MosonyOgawa}).
\end{remark}

\section{Exponent in smoothing the max-relative entropy}
  \label{sec:smooth-entropy}
The max-relative entropy is defined in Eq.~(\ref{equ:maxdefinition}). The smoothed version based on the purified distance is given by the following definition~\cite{Datta2009min}.
\begin{definition}
Let $\rho \in \mc{S}(\mc{H})$, $\sigma \in \mc{P}(\mc{H})$,
and $0\leq\epsilon<1$. The smooth max-relative entropy is defined as
\[
D^\epsilon_{\rm{max}}(\rho \| \sigma):=\min_{\tilde{\rho} \in \mc{B}
^{\epsilon}(\rho)} D_{\rm{max}}(\tilde{\rho} \| \sigma).
\]
\end{definition}

In this section, we investigate the asymptotic behavior of the exponential
decay of the small modification in smoothing the max-relative entropy. To
formulate the problem in an equivalent way, we define the smoothing quantity,
for any
$\rho\in\mc{S}(\mc{H})$, $\sigma \in \mc{P}(\mc{H})$ and $\lambda\in\mathbb{R}$,
\begin{equation}
  \label{eq:def-err}
\epsilon(\rho\|\sigma,\lambda):= \min\left\{\epsilon \;|\; D^{\epsilon}_{\rm{max}}(\rho \| \sigma) \leq \lambda  \right\}=\min\left\{P(\rho,\tilde{\rho}) \;|\; \tilde{\rho} \leq
2^\lambda\sigma \quad\text{and}\quad \tilde{\rho} \in \mc{S}_{\leq}(\mc{H}) \right\}.
\end{equation}
We determine the precise exponential rate of decay for $\epsilon(\rho^
{\ox n}\|\sigma^{\ox n}, nr)$.

\begin{theorem}
 \label{theorem:exp-mre}
 For arbitrary $\rho\in\mc{S}(\mc{H})$, $\sigma \in \mc{P}
 (\mc{H})$, and $r\in\mathbb{R}$, we have
 \begin{equation}
 \label{eq:max}
  \lim_{n\rightarrow\infty}\frac{-1}{n}\log \epsilon(\rho^{\ox n}\|\sigma^{\ox n},nr)
  =\frac{1}{2}\sup_{s\geq 0}\big\{s\big(r-D_{1+s}(\rho\|\sigma)\big)\big\}.
 \end{equation}
\end{theorem}

The above theorem can be rewritten as follows. There exists a sequence $\varepsilon_n\to 0$
such that
\begin{align}
D^{2^{-n (r_e+\varepsilon_n)}}_{\rm{max}}(\rho^{\ox n} \| \sigma^{\ox n})
=nr\label{H1}
\end{align}
with $r_e= \frac{1}{2}\sup_{s\geq 0}\big\{s\big(r-D_{1+s}(\rho\|\sigma)\big)\big\}$. When $r\leq D(\rho\|\sigma)$, the right hand side of Eq.~\eqref{eq:max} is zero. Otherwise, it is strictly positive.

The quantum asymptotic equipartition property~\cite{TCR2009fully, Tomamichel2015quantum} states that, as $n\rar\infty$, $\epsilon(\rho^{\ox n}\|\sigma^{\ox n},nr)\rar 0$ when $r>D(\rho\|\sigma)$ and $\epsilon(\rho^{\ox n}\|\sigma^{\ox n},nr)\rar 1$ when $r<D(\rho\|\sigma)$. Moreover, these convergences are exponentially fast. Our result of Theorem~\ref{theorem:exp-mre} has provided the exact exponent for the decay of $\epsilon(\rho^{\ox n}\|\sigma^{\ox n},nr)$ in the case $r>D(\rho\|\sigma)$. This is in analogy
to the Hoeffding bound~\cite{Nagaoka2006converse, Hayashi2007error, ANSV2008asymptotic} for the hypothesis testing relative entropy.

\begin{proofof}{Theorem~\ref{theorem:exp-mre}}
At first, we deal with the "$\geq$" part. This is done by deriving a general
upper bound for $\epsilon(\rho\|\sigma,\lambda)$, and then we apply it to the
asymptotic situation. Set
\begin{equation}
  \label{eq:mre-proj}
  Q:=\big\{\mc{E}_\sigma(\rho) \leq \frac{1}{v(\sigma)}2^\lambda\sigma\big\},
\end{equation}
where $\mc{E}_\sigma$ is the pinching map and $v(\sigma)$ is the number of
distinct eigenvalues of $\sigma$. We consider the state $\tilde{\rho}=Q\rho Q$.
On the one hand, by the pinching inequality~(\ref{eq:pinchingineq}) and the
definition of $Q$, we have
\begin{equation}\begin{split}
  Q\rho Q &\leq v(\sigma) Q\mc{E}_\sigma(\rho)Q  \\
          &\leq v(\sigma) Q\left(\frac{1}{v(\sigma)}2^\lambda\sigma\right) Q  \\
          &\leq 2^\lambda\sigma.       \label{eq:mre-proof-a}
\end{split}\end{equation}
On the other hand, we can bound the distance between $\rho$ and $\tilde{\rho}$
as follows. Firstly,
\[\begin{split}
  P(\rho,\tilde{\rho}) & =   \sqrt{1-F(\rho, Q\rho Q)^2} \\
                       & =   \sqrt{1-(\tr\rho Q)^2}  \\
                       &\leq \sqrt{2\tr\rho(\1-Q)}.
\end{split}\]
Then, denoting $p=\tr\rho(\1-Q)$ and $q=\tr\sigma(\1-Q)$, from the definition
of $Q$ we easily see that $p\geq \frac{1}{v(\sigma)}2^\lambda q$. So, for any $s\geq 0$,
\begin{equation}\begin{split}
  P(\rho,\tilde{\rho})
&\le \sqrt{2p^{1+s}p ^{-s}}
  \leq \sqrt{2\left(p^{1+s} \big(\frac{1}{v(\sigma)}2^\lambda q \big)^{-s} \right)}  \\
  &\leq \sqrt{2\left(p^{1+s} \big(\frac{1}{v(\sigma)}2^\lambda q \big)^{-s} +(1-p)^{1+s}
   \big(\frac{1}{v(\sigma)}2^\lambda (\tr\sigma-q) \big)^{-s}\right)}  \\
  & =   \sqrt{2v(\sigma)^{s}\,2^{s\big(D_{1+s}((p,1-p)\|(q,\tr\sigma-q))-\lambda\big)}} \\
  &\leq \sqrt{2v(\sigma)^{s}\,2^ {s\big(D_{1+s}(\rho\|\sigma)-\lambda\big)}},
    \label{eq:mre-proof-b}
\end{split}\end{equation}
where the last line is by the data processing inequality for the sandwiched
R{\'e}nyi divergence under quantum measurements (Proposition~
\ref{prop:srd-properties} (\romannumeral3)). Eq.~(\ref{eq:mre-proof-a}) and
Eq.~(\ref{eq:mre-proof-b}) imply that
\[
\epsilon(\rho\|\sigma,\lambda) \leq \sqrt{2v(\sigma)^{s}\,2^ {s\big(D_{1+s}
                                             (\rho\|\sigma)-\lambda\big)}}.
\]
This further gives
\begin{equation}
  \label{eq:mre-proof-low}
\liminf_{n\rightarrow\infty} \frac{-1}{n}\log \epsilon(\rho^{\ox n}\|
       \sigma^{\ox n}, nr)
\geq \frac{1}{2}\sup_{s\geq 0}\big\{s\big(r-D_{1+s}(\rho\|\sigma)\big)\big\}.
\end{equation}
Here we have also used the inequality $v(\sigma^{\ox n})\leq (n+1)^{
\operatorname{rank}(\sigma)}$ (see, e.g.~\cite{CoverThomas1991elements},
Theorem 12.1.1).

Next, we turn to the derivation of the other direction. Let
$\rho_n\in\mc{S}_\leq (\mc{H}^{\ox n})$ be any subnormalized
state which satisfies
\begin{equation}
  \label{eq:mre-proof-c}
  \rho_n \leq 2^{nr}\sigma^{\ox n}.
\end{equation}
We are to lower bound the purified distance between $\rho^{\ox n}$ and
$\rho_n$. Set $Q_n:=\{\rho^{\ox n}>9\cdot 2^{nr} \sigma^{\ox n}\}$. Denote
$p_n=\tr\rho^{\ox n}Q_n$ and $q_n=\tr\rho_n Q_n$, which are the probabilities
of obtaining the outcome associated with $Q_n$ when a projective measurement
$\{Q_n, \1-Q_n\}$ is applied to $\rho^{\ox n}$ and $\rho_n$, respectively.
Then, by Eq.~(\ref{eq:mre-proof-c}) and the definition of $Q_n$, it is easy
to see that
\[\begin{split}
  Q_n\rho^{\ox n}Q_n &\geq 9\cdot 2^{nr} Q_n\sigma^{\ox n}Q_n \\
                     &\geq 9 Q_n\rho_n Q_n,
\end{split}\]
which gives
\begin{equation}
\label{eq:mre-proof-d}
p_n \geq 9q_n.
\end{equation}
Now by the monotonicity of the fidelity under quantum measurements, we have
\begin{equation*}
\begin{split}
F(\rho^{\ox n}, \rho_n) &\leq F\big((p_n,1-p_n),(q_n,\tr \rho_n-q_n)\big)\\
                        &\leq \sqrt{p_n} \sqrt{q_n}+\sqrt{1-p_n}\\
                        &\leq \frac{p_n}{3}+\sqrt{1-p_n},
\end{split}
\end{equation*}
where for the last line Eq.~(\ref{eq:mre-proof-d}) is used. Thus,
\begin{equation*}
\begin{split}
P(\rho^{\ox n},\rho_n)& =    \sqrt{1-F^2(\rho^{\ox n},\rho_n)} \\
& \ge
\sqrt{
1- \Big(\frac{p_n}{3}+\sqrt{1-p_n}\Big)^2
}\\
& =
\sqrt{
-\frac{p_n^2}{9}+ p_n
-\frac{2p_n}{3}
\sqrt{1-p_n}
}\\
& \ge
\sqrt{p_n}\sqrt{
-\frac{p_n}{9}+ 1
-\frac{2}{3}
}\\
& =\sqrt{p_n}\sqrt{\frac{1}{3}-\frac{p_n}{9}}.
\end{split}
\end{equation*}
Because $\rho_n$ is an arbitrary subnormalized state that satisfies
Eq.~(\ref{eq:mre-proof-c}), we obtain
\begin{equation}
\label{eq:mre-proof-e}
     \epsilon(\rho^{\ox n}\| \sigma^{\ox n}, nr)
\geq \sqrt{p_n}\sqrt{\frac{1}{3}-\frac{p_n}{9}}.
\end{equation}
Proposition~\ref{prop:MosonyiOgawa} provides the exact rate of exponential
 decay for $p_n$ in (\ref{eq:mre-proof-e}), yeilding
\begin{equation}
  \label{eq:mre-proof-up}
\limsup_{n\rightarrow\infty}\frac{-1}{n}\log\epsilon(\rho^{\ox n}\|\sigma^{\ox n},nr)
\leq \frac{1}{2}\sup_{s\geq 0}\big\{s\big(r-D_{1+s}(\rho\|\sigma)\big)\big\}.
\end{equation}
Combining Eq.~\eqref{eq:mre-proof-low} and Eq.~\eqref{eq:mre-proof-up} we complete
the proof.
\end{proofof}

\begin{remark}
For the first part (the "$\geq$" part) of the proof of Theorem~\ref{theorem:exp-mre},
we can also employ the method introduced in~\cite{DattaRenner2009smooth}
(cf. Lemma 7 and Lemma 8) to construct the state $\tilde{\rho}$. This method was
later used and refined in~\cite{TCR2009fully} and~\cite{Tomamichel2015quantum},
yielding tight upper bound for $\epsilon(\rho\|\sigma,\lambda)$. Our approach here is
more direct. However, the price to pay is that an additional quantity $v(\sigma)$
is involved.
\end{remark}

\section{Security exponent of privacy amplification against quantum adversaries}
  \label{sec:privacy-amp}
Assume that two parties, Alice and Bob, share some
common classical randomness, represented by a random variable $X$ which takes
any value $x\in\mc{X}$ with probability $p_x$. The information of $X$ is
partially leaked to an adversary Eve, and is stored in a quantum system $E$
whose state is correlated with $X$. This situation is described by the following
classical-quantum (CQ) state
\begin{equation}
  \label{eq:source-state}
\rho_{XE}=\sum_{x} p_x \proj{x}_X \ox \rho^x_E.
\end{equation}

In the procedure of privacy amplification, Alice and Bob apply a hash function
$f:\mc{X}\rightarrow\mc{Z}$ to extract a random number $Z$, which is expected
to be uniformly distributed and independent of the adversary's system $E$. This
results in the state
\begin{equation}
  \label{eq:final-state}
\rho_{ZE}^f:=\sum_{z} \proj{z}_Z \ox \sum_{x\in f^{-1}(z)} p_x\rho^x_E
\end{equation}
on systems $Z$ and $E$. The size of the extracted randomness is $|\mc{Z}|$ and
the security is measured
by the closeness of this real state to the ideal state $\frac{\1_Z}{|\mc{Z}|}
\ox\rho_E$. In this paper, we consider two security measures, the insecurity
$P(\rho_{ZE}^f,\frac{\1_Z}{|\mc{Z}|}\ox\rho_E)$ in terms of purified distance,
and the insecurity $D(\rho_{ZE}^f \| \frac{\1_Z}{|\mc{Z}|}\ox\rho_E)$
in terms of relative entropy. These two measures have been extensively used
in the literature for privacy amplification. See,
e.g.,~\cite{TomamichelHayashi2013hierarchy, ABJT2020partially} for the purified
distance measure, and~\cite{BBCM1995generalized, Hayashi2011exponential,
Hayashi2015precise} for the relative entropy measure.
The latter is also called modified quantum mutual information and is related
to the leaked information~\cite{Hayashi2015precise}. Since it can be written
as
\[\begin{split}
      D(\rho_{ZE}^f \| \frac{\1_Z}{|\mc{Z}|}\ox\rho_E)
     &=I(Z;E)_{\rho^f}+D(\rho_Z^f \| \frac{\1_Z}{|\mc{Z}|}) \\
     &=\log |\mc{Z}| - H(Z|E)_{\rho^f},
\end{split}\]
we can understand it as the leaked information plus the nonuniform of the
extracted randomness, or the difference between the ideal ignorance and the
real ignorance of the extracted randomness, from the viewpoint of the
adversary.

The two-universal family of hash functions are commonly employed to extract
private randomness. It has the advantage of being universal (irrelevant of the
detailed structure of the state $\rho_{XE}$), as well as being efficiently
realizable~\cite{CarterWegman1979universal,BBCM1995generalized,Renner2005security,
TSSR2011leftover, Hayashi2015precise}. This is particularly useful in the
cryptographic setting. Let $\mc{F}$ be a set of hash functions from $\mc{X}$
to $\mc{Z}$, and $F$ represent a random choice of hash function $f$ from (a
subset of) $\mc{F}$ with probability $P_F(f)$. If $\forall (x_1,x_2)\in\mc{X}^2$
with $x_1\neq x_2$,
\begin{equation}
  \label{eq:two-universal}
\Pr\big\{F(x_1)=F(x_2)\big\} \leq \frac{1}{|\mc{Z}|},
\end{equation}
we say that the pair $(\mc{F},P_F)$ is two-universal, and that $F$ is
a two-universal random hash function.

The preceding work \cite{Hayashi2015precise} has derived an upper bound, in terms of the sandwiched R{\'e}nyi
divergence, for the insecurity of privacy amplification under the relative
entropy measure. When $n$-multiple copies of the
state~(\ref{eq:source-state}) are available, this provides an achievable
rate of the exponential decreasing of the insecurity, when the number of
copies $n$ increase. We are interested in the problem of determining the
precise exponent under which the insecurity decreases.

\subsection{Main results}
  \label{subsec:main-results}
At first, we derive a general upper bound for the rate of exponential
decreasing of the insecurity in privacy amplification, under both the
purified distance measure and the relative entropy measure.

\begin{theorem}
  \label{theorem:exp-privacy-up}
Let $\rho_{XE}$ be a CQ state, $\mc{F}_{n}(R)$ be the set of functions
from $\mc{X}^n$ to $\mc{Z}_n=\{1,\ldots,2^{nR} \}$. Let $\rho^{f_n}_{Z_nE^n}$
denote the state resulting from applying a hash function $f_n\in\mc{F}
_{n}(R)$ to $\rho_{XE}^{\ox n}$. For any fixed randomness extraction rate
$R \geq 0$, we have
 \begin{align}
   \limsup_{n\rightarrow\infty} \frac{-1}{n}\log \min_{f_n \in \mc{F}_n(R)}
   P(\rho^{f_n}_{Z_nE^n},\frac{\1_{Z_n}}{\lvert\mc{Z}_n\rvert} \ox {\rho^{\ox n}_E})
 & \leq\frac{1}{2}\sup_{s \geq 0} \big\{s\big(H_{1+s}(X|E)_{\rho}-R\big)\big\},
                                                  \label{eq:exp-privacy-up-P} \\
   \limsup_{n\rightarrow\infty} \frac{-1}{n}\log \min_{f_n \in \mc{F}_n(R)}
   D(\rho^{f_n}_{Z_nE^n}\|\frac{\1_{Z_n}}{\lvert\mc{Z}_n\rvert}\ox {\rho^{\ox n}_E})
 & \leq\sup_{s \geq 0} \big\{s\big(H_{1+s}(X|E)_{\rho}-R\big)\big\}.
                                                  \label{eq:exp-privacy-up-D}
 \end{align}
\end{theorem}

\begin{remark}
The work~\cite{HayashiTan2016equivocations} have proved
Eq.~(\ref{eq:exp-privacy-up-D}) in the classical case, where $\rho_{XE}$ is
fully classical. Theorem~\ref{theorem:exp-privacy-up} has extended this result
to the quantum setting.
\end{remark}

By combining Theorem \ref{theorem:exp-privacy-up} and a lower bound derived in~\cite{Hayashi2015precise}, we can get the exact exponent of the asymptotic decreasing of the
insecurity when the rate of randomness extraction is above a critical value.
\begin{theorem}
 \label{theorem:exp-privacy}
Let $\rho_{XE}$ be a CQ state, $\mc{F}_{n}(R)$ be the set of functions
from $\mc{X}^n$ to $\mc{Z}_n=\{1,\ldots,2^{nR} \}$, $F_n$ be any two-universal
random hash function drawn from (a subset of) $\mc{F}_{n}(R)$,
and $R_{\rm{critical}}:=\frac{\mathrm{d}}{\mathrm{d} s} sH_{1+s}(X|E)_{\rho}\big|_{s=1}$. For the
rate $R$ of randomness extraction satisfying $R\geq R_{\rm{critical}}$,
we have
\begin{equation}\begin{split}
 \label{eq:exp-privacy-P}
 \lim_{n\rightarrow\infty} \frac{-1}{n}\log \min_{f_n \in \mc{F}_n(R)}
 P(\rho^{f_n}_{Z_nE^n},\frac{\1_{Z_n}}{\lvert\mc{Z}_n\rvert}\ox{\rho_E^{\ox n}})
 &=\lim_{n\rightarrow\infty}\frac{-1}{n}\log\mathbb{E}_{F_n}P(\rho^{F_n}_{Z_nE^n},
   \frac{\1_{Z_n}}{\lvert \mc{Z}_n \rvert} \ox {\rho_E^{\ox n}})  \\
 &=\frac{1}{2}\max_{0\leq s\leq 1} \big\{s\big(H_{1+s}(X|E)_{\rho}-R\big)\big\},
\end{split}\end{equation}
\begin{equation}\begin{split}
 \label{eq:exp-privacy-D}
 \lim_{n\rightarrow\infty} \frac{-1}{n} \log \min_{f_n \in \mc{F}_n(R)}
 D(\rho^{f_n}_{Z_nE^n}\|\frac{\1_{Z_n}}{\lvert\mc{Z}_n\rvert}\ox{\rho_E^{\ox n}})
 &=\lim_{n\rightarrow\infty}\frac{-1}{n}\log\mathbb{E}_{F_n}D(\rho^{F_n}_{Z_nE^n}
   \| \frac{\1_{Z_n}}{\lvert \mc{Z}_n \rvert} \ox {\rho_E^{\ox n}}) \\
 &=\max_{0\leq s\leq 1} \big\{s\big(H_{1+s}(X|E)_{\rho}-R\big)\big\}.
\end{split}\end{equation}
\end{theorem}

The proof of Theorem~\ref{theorem:exp-privacy-up} is based on the result
obtained in Section~\ref{sec:smooth-entropy} on the exponent in smoothing
the max-relative entropy. To relate privacy amplification to the smooth
max-relative entropy in a proper way,  we employ a version of the smooth
conditional min-entropy~\cite{TSSR2011leftover, ABJT2020partially}.
\begin{definition}\label{def:scme}
For a state $\rho_{AB} \in \mc{S}(\mc{H}_{AB})$, the smooth
conditional min-entropy is defined as
\begin{equation}
    H^{\epsilon}_{\rm {min}}(A|B)_\rho
 := -D_{\rm{max}}^\epsilon(\rho_{AB} \| \1_A \ox \rho_B).
\end{equation}
\end{definition}
When $\epsilon=0$, we recover the (non-smoothed) conditional min-entropy
$H_{\rm {min}}(A|B)_\rho:= -D_{\rm{max}}(\rho_{AB} \| \1_A \ox \rho_B)$.

\begin{proposition}
\label{prop:monohash}
Let $\sigma_{XAB}=\sum_{x} p_x \proj{x}_X \ox \sigma^x_{AB}$ be a state in
$\mc{S}(\mc{H}_{XAB})$. Let $f:\mc{X}\rightarrow\mc{Z}$ be a function
and let $Z=f(X)$. Then,
\[
\begin{split}
             H^{\epsilon}_{\rm{min}}(XA|B)_\sigma
            &\geq H^{\epsilon}_{\rm{min}}(ZA|B)_\sigma,\ \text{where} \\
\sigma_{ZAB}&=\sum_{z} \proj{z}_Z \ox\big(\sum_{x \in f^{-1}(z)}p_x
\sigma^{x}_{AB}\big).
\end{split}
\]
\end{proposition}

There is another definition of the smooth conditional min-entropy (see, e.g.,
\cite{TCR2009fully, TCR2010duality,TomamichelHayashi2013hierarchy}):
\begin{equation}
    \bar{H}^{\epsilon}_{\rm {min}}(A|B)_\rho
 := - \min_{\sigma_B} D_{\rm{max}}^\epsilon(\rho_{AB} \| \1_A \ox \sigma_B).
\label{min-hierarchy}
\end{equation}
Since the reference \cite[Proposition 3]{TomamichelHayashi2013hierarchy}
showed the same statement as Proposition \ref{prop:monohash}
under the different definition \eqref{min-hierarchy},
the proof of Proposition~\ref{prop:monohash} is analogous to the proof of
Proposition 3 in~\cite{TomamichelHayashi2013hierarchy} and is given in the
Appendix.
To see the relation between the smooth conditional min-entropy and the
insecurity, we show the following proposition.

\begin{proposition}\label{H2}
Let $\rho_{XE}$ be a CQ state. When $\log \lvert\mc{Z} \rvert \ge
H^{\epsilon}_{\rm {min}}(X|E)_\rho$, any function $f:\mc{X}\rightarrow\mc{Z}$
satisfies
\begin{align}
   P(\rho^{f}_{Z E },\frac{\1_{Z}}{\lvert\mc{Z}\rvert} \ox {\rho_E})
 \geq \epsilon, \label{XPO}
\end{align}
where $\rho_{ZE}^f$ is a state of the form~(\ref{eq:final-state})
resulting from applying $f$ to $\rho_{XE}$.
\end{proposition}
In fact, the reference \cite[Theorem 8]{TomamichelHayashi2013hierarchy}
showed the same statement as Proposition \ref{H2}
under the different definition \eqref{min-hierarchy}.
Hence, it can be shown in a similar way.

\medskip
\begin{proofof}{Proposition \ref{H2}}
For any function $f:\mc{X}\rightarrow\mc{Z}$, Proposition~\ref{prop:monohash}
applies, giving
\beq\label{eq:privup-1}
  H^\epsilon_\text{min}(X|E)_\rho \geq H^\epsilon_\text{min}(Z|E)_{\rho^f}.
\eeq
We choose $\epsilon'$ such that $H^{\epsilon'}_\text{min}(Z|E)_{\rho^f}
=\log \lvert\mc{Z} \rvert$. By the definition of the smooth conditional
min-entropy, we find that
\begin{align}
   P(\rho^{f}_{Z E },\frac{\1_{Z}}{\lvert\mc{Z}\rvert} \ox {\rho_E})
\ge \epsilon'.
\end{align}
Also, Eq.~\eqref{eq:privup-1} implies
$\epsilon' \ge \epsilon$. Therefore, we obtain~\eqref{XPO}.
\end{proofof}

Now, we are ready to prove Theorem~\ref{theorem:exp-privacy-up} and Theorem~\ref{theorem:exp-privacy}.

\medskip
\begin{proofof}{Theorem~\ref{theorem:exp-privacy-up}}
Eq.~(\ref{eq:exp-privacy-up-P}) can be shown by the combination of
Theorem~\ref{theorem:exp-mre} and Proposition \ref{H2} as follows.
We choose $r_e:=\frac{1}{2}\sup_{s \geq 0} \big\{s\big(H_{1+s}(X|E)_{\rho}-R\big)\big\}$.
Eq.~\eqref{H1}, i.e.,
Theorem~\ref{theorem:exp-mre} guarantees the existence of a sequence
$\varepsilon_n\to 0$ such that
$nR = H^{2^{-n (r_e+\varepsilon_n)}}_{\rm {min}}(X^n|E^n)_{\rho_{XE}^{\ox n}}$.
Hence, Proposition \ref{H2} guarantees
 \begin{align*}
   &\limsup_{n\rightarrow\infty} \frac{-1}{n}\log \min_{f_n \in \mc{F}_n(R)}
   P(\rho^{f_n}_{Z_nE^n},\frac{\1_{Z_n}}{\lvert\mc{Z}_n\rvert} \ox {\rho^{\ox n}_E})\\
 \leq & \lim_{n\rar\infty} \frac{-1}{n}\log 2^{-n (r_e+\varepsilon_n)} \\
   =  &r_e,
 \end{align*}
which coincides with Eq.~(\ref{eq:exp-privacy-up-P}).

To prove Eq.~(\ref{eq:exp-privacy-up-D}), we make use of a relation between
the relative entropy and the purified distance. By definition, we easily
see that
\[
  D_{\frac{1}{2}}(\rho \| \sigma)=-2\log F(\rho, \sigma).
\]
Meanwhile, since $D_\alpha$ is nondecreasing with $\alpha$,
\[
  D_{\frac{1}{2}}(\rho \| \sigma) \leq D(\rho \| \sigma).
\]
Thus,
\beq\label{eq:privup-6}
P(\rho, \sigma)=\sqrt{1-F^2(\rho,\sigma)}
               \leq\sqrt{1-2^{-D(\rho \| \sigma)}}
               \leq\sqrt{(\ln 2)D(\rho \| \sigma)}.
\eeq
Eq.~(\ref{eq:exp-privacy-up-D}) follows directly from Eq.~(\ref{eq:privup-6})
and Eq.~(\ref{eq:exp-privacy-up-P}), and we complete the proof.
\end{proofof}

\medskip
\begin{proofof}{Theorem~\ref{theorem:exp-privacy}}
The preceding work \cite[Theorem 1]{Hayashi2015precise} has proved that under the conditions of
Theorem~\ref{theorem:exp-privacy-up} and for any two-universal hash function
$F_n$ drawn from (a subset of) $\mc{F}_{n}(R)$,
\begin{equation}\begin{split}
 \label{eq:exp-privacy-1}
 \liminf_{n\rightarrow\infty} \frac{-1}{n} \log \min_{f_n \in \mc{F}_n(R)}
  D(\rho^{f_n}_{Z_nE^n}\|\frac{\1_{Z_n}}{\lvert\mc{Z}_n\rvert}\ox{\rho_E^{\ox n}})
 &\geq\liminf_{n\rightarrow\infty}\frac{-1}{n}\log\mathbb{E}_{F_n}D(\rho^{F_n}_
      {Z_nE^n}\| \frac{\1_{Z_n}}{\lvert \mc{Z}_n \rvert} \ox {\rho_E^{\ox n}}) \\
 &\geq\max_{0\leq s \leq 1} \big\{s\big(H_{1+s}(X|E)_{\rho}-R\big)\big\}.
\end{split}\end{equation}
Making use of Eq.~(\ref{eq:privup-6}) and the concavity of the square root function,
we are able to get a similar bound for the purified distance measure from
Eq.~(\ref{eq:exp-privacy-1}), under the same conditions. Namely,
\begin{equation}\begin{split}
 \label{eq:exp-privacy-2}
 \liminf_{n\rightarrow\infty} \frac{-1}{n} \log \min_{f_n \in \mc{F}_n(R)}
  P(\rho^{f_n}_{Z_nE^n},\frac{\1_{Z_n}}{\lvert\mc{Z}_n\rvert}\ox{\rho_E^{\ox n}})
 &\geq\liminf_{n\rightarrow\infty}\frac{-1}{n}\log\mathbb{E}_{F_n}P(\rho^{F_n}_
      {Z_nE^n}, \frac{\1_{Z_n}}{\lvert \mc{Z}_n \rvert} \ox {\rho_E^{\ox n}}) \\
 &\geq\liminf_{n\rightarrow\infty}\frac{-1}{n}\log\sqrt{(\ln 2) \mathbb{E}_{F_n}D(\rho^{F_n}_
      {Z_nE^n}\| \frac{\1_{Z_n}}{\lvert \mc{Z}_n \rvert} \ox {\rho_E^{\ox n}})} \\
 &\geq\frac{1}{2}\max\limits_{0\leq s\leq 1}\big\{s\big(H_{1+s}(X|E)_{\rho}-R\big)\big\}.
\end{split}\end{equation}
If the lower bounds in Eq.~(\ref{eq:exp-privacy-1}) and Eq.~(\ref{eq:exp-privacy-2})
equal the upper bounds in Eq.~(\ref{eq:exp-privacy-up-D}) and
Eq.~(\ref{eq:exp-privacy-up-P}), respectively, we would obtain the exact rates
of exponential decay. In the following, we prove that this is indeed the case when $R \geq
R_{\rm{critical}}$.

Consider the optimization problem
\beq\label{eq:exp-privacy-sup}
E_u(R):=\sup_{s \geq 0} \big\{s\big(H_{1+s}(X|E)_{\rho}-R\big)\big\}.
\eeq
Since the function $s\mapsto sH_{1+s}(X|E)_{\rho}$ is concave (cf. Proposition~\ref{prop:srd-properties} (\romannumeral4)) and obviously continuously
differentiable on $(0,\infty)$, $s(H_{1+s}(X|E)_{\rho}-R)$ is also concave and
continuously differentiable as a function of $s$. So the supremum in
Eq.~(\ref{eq:exp-privacy-sup}) is achieved at the point with zero derivative
(if it exists), given by the solution of the equation
\beq\label{eq:exp-privacy-3}
R=\hat{R}(s)\equiv\frac{\mathrm{d}}{\mathrm{d} s} sH_{1+s}(X|E)_{\rho}.
\eeq
Note that the critical rate is
\beq\label{eq:exp-privacy-4}
  R_{\rm critical}=\hat{R}(1)
 \equiv\frac{\mathrm{d}}{\mathrm{d} s} sH_{1+s}(X|E)_{\rho}\big|_{s=1}.
\eeq
$\hat{R}(s)$ is nonincreasing, because $s\mapsto sH_{1+s}(X|E)_{\rho}$ is
concave. Also, we define 
\begin{align}
\hat{R}(0)&:=\ \,\lim_{s\rightarrow 0}\ \, \hat{R}(s) = H(X|E)_\rho,
  \label{eq:limit-0} \\
\hat{R}(\infty)&:=\lim_{s\rightarrow +\infty}\hat{R}(s)=H_\text{min}(X|E)_\rho,
  \label{eq:limit-infty}
\end{align}
where \eqref{eq:limit-0} and \eqref{eq:limit-infty} follow from
(51) and (52) of \cite[Lemma IV.2]{MosonyiOgawa2015quantum}, respectively.
There are four cases:
\begin{enumerate}[(i)]
  \item $R\geq H(X|E)_\rho$: the function $s\mapsto s(H_{1+s}(X|E)_{\rho}-R)$
        is monotonically dereasing. So the supremum in Eq.~(\ref{eq:exp-privacy-sup})
        is 0, achieved at $s=0$;
  \item $R_{\rm critical}\leq R< H(X|E)_\rho$: Eq.~(\ref{eq:exp-privacy-3})
        has a solution $s^*\in(0,1]$, where Eq.~(\ref{eq:exp-privacy-sup})
        achieves the supremum;
  \item $H_\text{min}(X|E)_\rho<R<R_{\rm critical}$: Eq.~(\ref{eq:exp-privacy-3})
        has a solution $s^*\in(1,+\infty)$, where Eq.~(\ref{eq:exp-privacy-sup})
        achieves the supremum;
  \item $R\leq H_\text{min}(X|E)_\rho$: the function $s\mapsto s(H_{1+s}(X|E)
        _{\rho}-R)$ is monotonically increasing. So the supremum in
        Eq.~(\ref{eq:exp-privacy-sup}) is $+\infty$, approached when $s\rar
        +\infty$.
\end{enumerate}
In cases (\romannumeral1) and (\romannumeral2),
we have that the supremum in Eq.~(\ref{eq:exp-privacy-sup}) is achieved at
$s\in[0,1]$. Therefore, the bound in Eq.~(\ref{eq:exp-privacy-up-D}) and
that in Eq.~(\ref{eq:exp-privacy-1}) are equal, and so are the bound in
Eq.~(\ref{eq:exp-privacy-up-P}) and that in Eq.~(\ref{eq:exp-privacy-2}).
Hence we complete the proof.
\end{proofof}

\medskip
Since $\hat{R}(s)$ is nonincreasing, $\hat{R}(s)$ has the inverse function $\psi$.
The results presented in Theorem~\ref{theorem:exp-privacy-up} and
Theorem~\ref{theorem:exp-privacy}
can be explained by using
$E_u(R)$ and $E_l(R):=\sup_{0\leq s\leq 1} {s(H_{1+s}(X|E)_{\rho}-R)}$
as follows.
\begin{align}\label{eq:H2}
E_u(R)
=&
\left\{
\begin{array}{ll}
0 & \hbox{ when } R \ge H(X|E)_{\rho}, \\
 \psi(R)H_{1+\psi(R)}(X|E)_{\rho}
- \psi(R) R & \hbox{ when }H(X|E)_{\rho} >  R \ge H_\text{min}(X|E)_\rho,\\
+\infty & \hbox{ when }H_\text{min}(X|E)_\rho >  R,
\end{array}
\right.
\\
E_l(R)=&
\left\{
\begin{array}{ll}
0 & \hbox{ when } R \ge H(X|E)_{\rho},  \\
 \psi(R)H_{1+\psi(R)}(X|E)_{\rho}- \psi(R) R & \hbox{ when } H(X|E)_{\rho} > R > R_{\rm critical}, \\
  H_{2}(X|E)_{\rho}- R & \hbox{ when } R_{\rm critical} \ge R .
\end{array}
\right.
\end{align}
Figure~\ref{fig:exponent} illustrates the above two functions.

We make a few remarks on a related security measure.
The quantity, $\min\limits_{\sigma_E}P(\rho_{ZE}^f,\frac{\1_Z}{|\mc{Z}|}\ox
\sigma_E)$, was employed in some works to measure the insecurity of the
extracted randomness $Z$ (see, e.g.,~\cite{TomamichelHayashi2013hierarchy}).
There is an additional minimization over the adversary's state, compared
to $P(\rho_{ZE}^f, \frac{\1_Z}{|\mc{Z}|}\ox\rho_E)$ that we use here.
Denoting the minimizer in that measure as $\sigma_E^*$, we have
\[\begin{split}
       P(\rho_{ZE}^f, \frac{\1_Z}{|\mc{Z}|}\ox\sigma_E^*)
&\leq  P(\rho_{ZE}^f, \frac{\1_Z}{|\mc{Z}|}\ox\rho_E)    \\
&\leq  P(\rho_{ZE}^f, \frac{\1_Z}{|\mc{Z}|}\ox\sigma_E^*)
      +P(\frac{\1_Z}{|\mc{Z}|}\ox\sigma_E^*, \frac{\1_Z}{|\mc{Z}|}\ox\rho_E)\\
&\leq 2P(\rho_{ZE}^f, \frac{\1_Z}{|\mc{Z}|}\ox\sigma_E^*).
\end{split}\]
So, there is no difference between these two measures regarding the
rate of asymptotic exponential decreasing. However, we prefer to employ
the measure $P(\rho_{ZE}^f, \frac{\1_Z}{|\mc{Z}|}\ox\rho_E)$ because
fixing $\rho_E$ in the measure fits better the requirement of composable
security (see discussions in \cite{PortmannRenner2014cryptographic} and
\cite{ABJT2020partially}).

\begin{figure}[ht]
  \includegraphics[width=9.5cm]{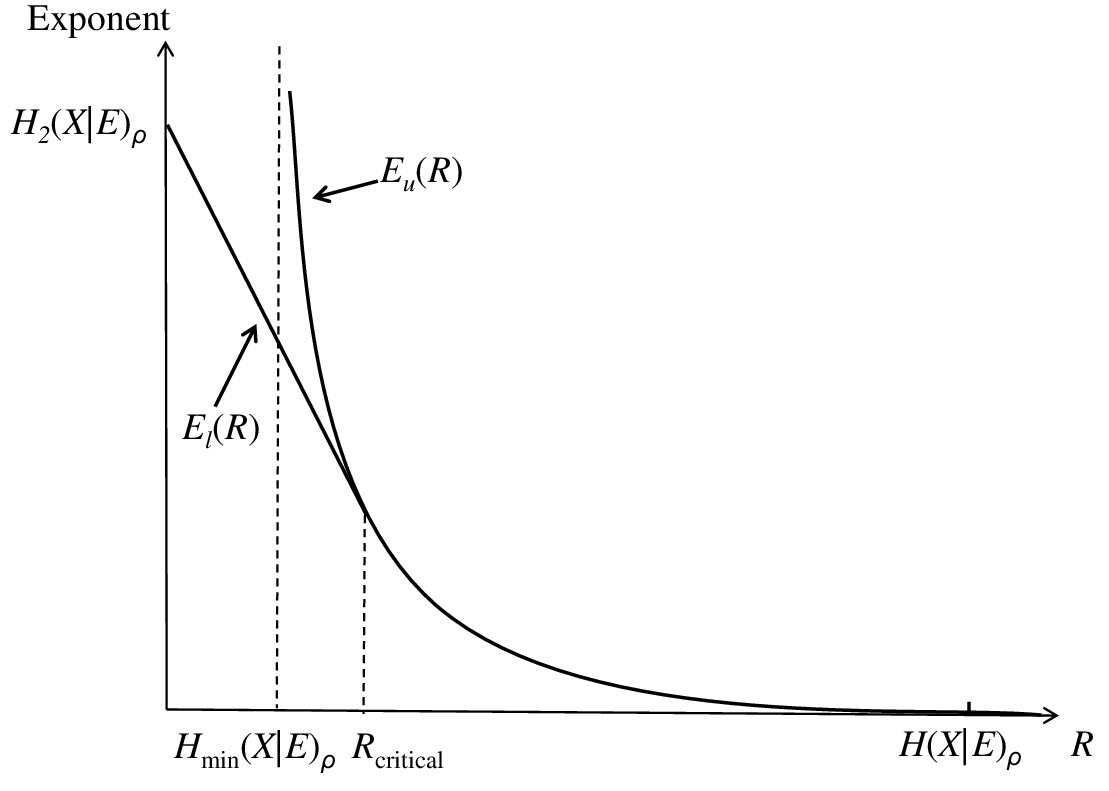} \centering
  \caption{Security exponent of privacy amplification.
  $E_u(R)$
  is the upper bound derived in the present paper.
  $E_l(R)$
  is the lower bound by the reference \cite{Hayashi2015precise}. These
  two bounds are equal when $R\geq R_\text{critical}$, giving
  the exact security exponent. When $R\geq H(X|E)_\rho$, the
  security exponent is $0$. Below the
  critical value $R_\text{critical}$, the upper bound $E_u(R)$
  is larger and diverges to infinity when $R<H_\text{min}(X|E)_\rho$,
  while the lower bound $E_l(R)$ becomes linear and reaches
  $H_2(X|E)_\rho$ at $R=0$.}
  \label{fig:exponent}
\end{figure}

\subsection{Discussion on the low-rate case}
  \label{subsec:low-rate}
In Theorem~\ref{theorem:exp-privacy}, we have obtained the exponents only when
$R \geq R_{\rm critical}$. One may guess that either the achievability bounds of Eq.~(\ref{eq:exp-privacy-1}) and Eq.~(\ref{eq:exp-privacy-2}) or
the converse bounds of Theorem~\ref{theorem:exp-privacy-up} are the exact exponents
when $R<R_{\rm critical}$. Here we
give two simple examples to show that this is not true, i.e., neither of them are
tight in general when $R<R_{\rm critical}$. This indicates that $R_{\rm critical}$
may be indeed a critical point in the exponential analysis of privacy amplification.

\smallskip
\noindent \textbf{Example 1} \enspace We consider the classical-quantum state $\rho_{XE}=(\frac{1}{3}\proj{0}+\frac{2}{3}\proj{1}) \ox \rho_E$.
We have
$s H_{1+s}(X|E)_{\rho}=\phi(s):=- \log ((\frac{1}{3})^{1+s}+(\frac{2}{3})^{1+s})$.
Then,  using the binary entropy
$h(x):= -x\log x -(1-x)\log (1-x)$, we have
$H(X|E)_{\rho}=h(\frac{1}{3})$ and
$\hat{R}(s)=\frac{\mathrm{d}}{\mathrm{d} s} sH_{1+s}(X|E)_{\rho}
=\frac{ (1+2^{1+s})\log 3- 2^{1+s}\log 2}{1+2^{1+s}}
$. In particular,
$R_{\rm critical}=\hat{R}(1)=
\frac{ 5 \log 3- 4 \log 2}{5}$ and
$\hat{R}(+\infty)=\log \frac{3}{2}=H_{\min}(X|E)_{\rho}$.
In addition,
$E_u(\log \frac{3}{2})$ is calculated as
\[
\begin{split}
E_u(\log \frac{3}{2})
=&\lim_{s \to +\infty}
\{- \log ((\frac{1}{3})^{1+s}+(\frac{2}{3})^{1+s})
-s \log \frac{3}{2}\} \\
=&\lim_{s \to +\infty}
\{\log \frac{3}{2}-\log (1+\frac{1}{2^{1+s}})
=\log \frac{3}{2}\}.
\end{split}
\]
Therefore, since $ H_{2}(X|E)_{\rho}=\log \frac{9}{5} $,
$E_u(R),E_l(R)$ are calculated as
\begin{align}\label{eq:H1}
E_u(R)
=&
\left\{
\begin{array}{ll}
0 & \hbox{ when } R \ge h(\frac{1}{3}), \\
\phi( \psi(R))- \psi(R) R & \hbox{ when }h(\frac{1}{3})>  R \ge \log \frac{3}{2},\\
+\infty  & \hbox{ when }\log \frac{3}{2}>  R,
\end{array}
\right.
\\
E_l(R)=&
\left\{
\begin{array}{ll}
0 & \hbox{ when } R \ge h(\frac{1}{3}),  \\
\phi( \psi(R))- \psi(R) R & \hbox{ when } h(\frac{1}{3})> R > \frac{ 5 \log 3- 4 \log 2}{5},\\
\log \frac{9}{5}- R & \hbox{ when } \frac{ 5 \log 3- 4 \log 2}{5} \ge R . 
\end{array}
\right.
\end{align}
Their behaviors are plotted as Fig. \ref{fig:exponent2}.
Notice that $\hat{R}(s)$ is strictly nonincreasing for $s$
because $s H_{1+s}(X|E)_{\rho}$ is a strictly concave function of $s$.
Hence,
we have $ \psi(R) > 1$ for $R < \frac{ 5 \log 3- 4 \log 2}{5}$.
Since $\frac{d (\phi( s)- s R)}{ds}|_{s=1}
=(\phi'( s)-R)|_{s=1} >0$,
$E_u(R)$ takes a larger value than $E_l(R)$ because
\begin{align}
E_u(R) = \phi( \psi(R))- \psi(R) R
> \phi( 1)- R= E_l(R).
\end{align}
Therefore, this case has the following three possible cases.
In the first case, $E_u(R)$ is the tight upper bound.
In the second case, $E_l(R)$ is the tight lower bound.
In the third case, neither $E_u(R)$ nor $E_l(R)$ is a tight bound.
To investigate this problem, we notice that
the eigenvalue
of $\rho_X^{\ox n}$ associated with the eigenvector $\ket{0,0,\cdots,0}$ is
$\frac{1}{3^n}$, and all the other eigenvalues are $\frac{1}{3^n}$ multiplied
by an even number. This simple fact will be crucial for our later estimation.

Let $f_n:\mc{X}^n\rar\mc{Z}_n$ be an arbitrary sequence of hash function (the
size $|\mc{Z}_n|$ is also arbitrary). Let $z_n^*=f_n(0,0,\cdots,0)$ and pick
$z'_n\in\mc{Z}_n$ such that $z'_n\neq z_n^*$. Then $\bra{z_n^*}\rho^{f_n}_{Z_n}
\ket{z_n^*}$ must be $\frac{1}{3^n}$ multiplied by a odd number and $\bra{z_n'}
\rho^{f_n}_{Z_n}\ket{z_n'}$ be $\frac{1}{3^n}$ multiplied by an even number. So
\begin{equation*}
\begin{split}
     &d(\rho^{f_n}_{Z_nE^n},\frac{\1_{Z_n}}{|\mathcal{Z}_n|} \ox \rho_E^{\ox n})\\
 =&    \frac{1}{2} \sum_{z_n \in \mathcal{Z}_n} \big|\bra{z_n}\rho^{f_n}_{Z_n}
       \ket{z_n}-\frac{1}{|\mathcal{Z}_n|}\big| \\
\geq & \frac{1}{2} \big( \big|\bra{z_n^*}\rho^{f_n}_{Z_n} \ket{z_n^*}-
       \frac{1}{|\mathcal{Z}_n|}\big| + \big|\bra{z'_n}\rho^{f_n}_{Z_n} \ket{z'_n}
                                              -\frac{1}{|\mathcal{Z}_n|}\big|\big)\\
\geq & \frac{1}{2} \big|\bra{z_n^*}\rho^{f_n}_{Z_n} \ket{z_n^*}-\bra{z'_n}
                                                 \rho^{f_n}_{Z_n} \ket{z'_n}\big| \\
\geq & \frac{1}{2 \times 3^n}.
\end{split}
\end{equation*}
With this in hand, the use of Pinsker's inequality and Fuchs-van de Graaf
inequality~\cite{FuchsVan1999cryptographic} leads respectively to
\begin{align}
  \label{eq:count-d}
\limsup_{n\rar\infty} -\frac{1}{n}\log\min_{f_n \in \mathcal{F}_n(R)}
   D(\rho^{f_n}_{\mathcal{Z}_nE_n} \| \frac{\1_{\mathcal{Z}_n}}{|\mathcal{Z}_n|}
      \ox \rho_E^{\ox n}) &\leq \log 9, \\
  \label{eq:count-p}
\limsup_{n\rar\infty} -\frac{1}{n}\log\min_{f_n \in \mathcal{F}_n(R)}
   P(\rho^{f_n}_{\mathcal{Z}_nE_n},\frac{\1_{\mathcal{Z}_n}}{|\mathcal{Z}_n|}
      \ox \rho_E^{\ox n}) &\leq \log 3,
\end{align}
for any randomness extraction rate $R>0$. Eq.~(\ref{eq:count-d}) and
Eq.~(\ref{eq:count-p}) also provide the same bounds for the exponents
in the average case where the insecurity is averaged over two-universal
hash functions.
On the other hand, for $R<H_{\text{min}}(X|E)_\rho=\log\frac{3}{2}$,
\eqref{eq:H1} shows that $E_u(R)=+\infty > \log 9=3.16993$.
Hence,
the upper bound $E_u(R)$
in Theorem~\ref{theorem:exp-privacy-up} is not the tight upper bound.
That is, for $R<H_{\text{min}}(X|E)_\rho$,
we have the second case or the third case.

\begin{figure}[ht]
  \includegraphics[width=9.5cm]{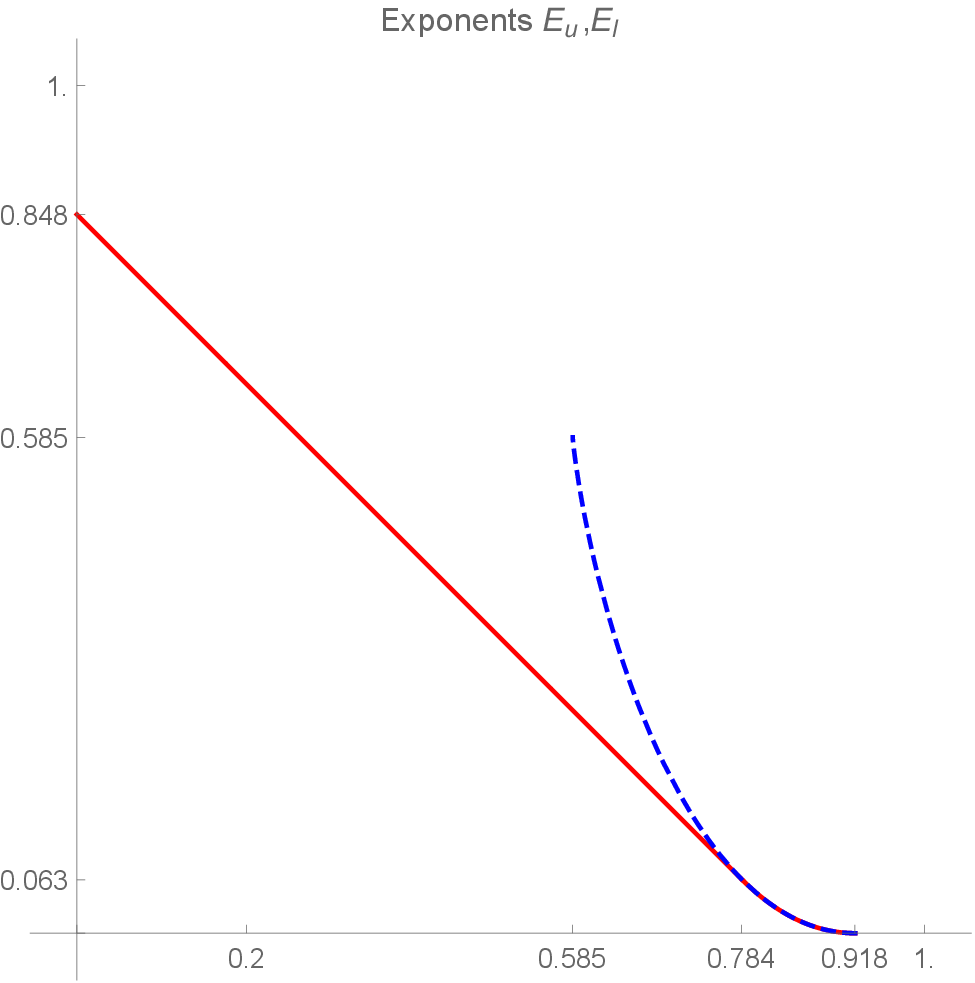} \centering
  \caption{Exponents in Example 1.
  Solid red curve expresses the lower bound  $E_l(R)$ by the reference \cite{Hayashi2015precise}.
  Dashed blue curve expresses the upper bound $E_u(R)$ derived in the present paper.
  These
  two bounds are equal when $R\geq \frac{ 5 \log 3- 4 \log 2}{5}=0.784963  $, giving
  the exact security exponent. When $R\geq 0.918296$,
  the security exponent is $0$. Below the
  critical value $0.784963 $, the upper bound $E_u(R)$
  is $\log \frac{3}{2}=0.584963$
   when $R=\log \frac{3}{2}$.
Since it diverges to infinity when $R<\log \frac{3}{2}$,
it is not plotted in this range.
  The lower bound $E_l(R)$ becomes linear and reaches
  $\log \frac{9}{5}=0.847997$ at $R=0$.}
  \label{fig:exponent2}
\end{figure}

\smallskip
\noindent \textbf{Example 2} \enspace Let $\rho_{XE}=(\frac{1}{4}\sum_
{i=1}^4\proj{i}) \ox \rho_E$, $\mathcal{Z}=\{0, 1\}$. We denote by $S_4$
the permutation group of $\mathcal{X}=\{1, 2, 3, 4\}$. Let $\Pi$ be the
random permutation over $\mathcal{X}$, i.e., it takes
the value $\pi \in S_4$ with equal probability for all $\pi$.
Define $f:\mathcal{X}\rar\mathcal{Z}$ by
\[
\label{equ:example2}
  f(i) := \begin{cases}
                      0 &  i=1, 2, \\
                      1 &  i=3, 4.
                    \end{cases}
\]
Then, we consider the random hash function $F_n:=(f \circ \Pi)^{\times n} $.
It is easy to see that $F_n$ is two-universal. But on the other hand,
it always holds that
\[
\rho^{F_n}_{Z_nE^n}=\frac{\1_{Z_n}}{|\mathcal{Z}_n|} \ox \rho_E^{\ox n},
\]
where $\mathcal{Z}_n=\mathcal{Z}^n$. Hence, $\mathbb{E}_{F_n} D(\rho^{F_n}_
{Z_nE^n} \| \frac{\1_{Z_n}}{|\mathcal{Z}_n|} \ox \rho_E^{\ox n})=
\mathbb{E}_{F_n}P(\rho^{F_n}_{Z_nE^n},\frac{\1_{Z_n}}{|\mathcal{Z}_n|}
\ox \rho_E^{\ox n})=0$, and the corresponding exponents are
$+\infty$. This is also true when the expectations are replaced by
the minimization over all hash functions from $\mc{X}^n$ to $\mc{Z}_n$. So,
the lower bounds of Eq.~(\ref{eq:exp-privacy-1}) and Eq.~(\ref{eq:exp-privacy-2}),
which are finite everywhere, are not tight in general.

\section{Asymptotic equivocation rate and security exponent under the sandwiched R\'enyi divergence}
  \label{sec:equivocation}
The equivocation rate is the adversary's maximum ambiguity rate for
a given randomness extraction rate $R$. Specifically, for a CQ state $\rho_{XE}$ and a randomness extraction rate $R$, the equivocation rate $\mathcal{R}_s(R|\rho)$ under the sandwiched R\'enyi divergence of order $1+s$ security measure is defined as
\[
\mathcal{R}_s(R|\rho):=\lim_{n \rightarrow \infty} \frac{1}{n}\max_{f_n} H_{1+s}(Z_n|E^n)_{\rho^{f_n}} ,
\]
where the maximization is taken over all maps $f_n : \mathcal{X}^{\times n} \rightarrow \mathcal{Z}_n$ and
$\rho^{f_n}_{Z_nE^n}$ is the state resulting from applying $f_n$ to $\rho_{XE}^{\ox n}$. In some papers, the equivocation rate is also defined as the adversary's minimum information rate for a given randomness extraction rate $R$, i.e.,
\[
\mathcal{R}'_s(R|\rho):=\lim_{n \rightarrow \infty} \frac{1}{n} \min_{f_n} D_{1+s}(\rho^{f_n}_{Z_nE^n} \| \frac{\1_{Z_n}}{|\mathcal{Z}_n|} \ox \rho_E^{\ox n}).
\]
These two definitions are related. Indeed, it is easy to see that
\[
\mathcal{R}_s(R|\rho)=R-\mathcal{R}'_s(R|\rho).
\]
In our paper, we take the second definition.

The concept of equivocation was first proposed
by Wyner~\cite{Wyner1975wire}
and was studied by many researchers in the wiretap scenario.
In the quantum privacy amplification scenario,
the preceding work~\cite{Hayashi2015precise} derived the equivocation rate
under the quantum relative entropy security measure.
Later, the reference~\cite{HayashiTan2016equivocations} derived the
equivocation rate, in the classical privacy amplification scenario, under
the R{\'e}nyi relative entropy security measure.

In this section, we investigate the asymptotic equivocation rate and
the security exponents under the sandwiched R\'enyi divergence security
measure, with R\'enyi parameter in $(1,2]$. This generalizes the results
by~\cite{HayashiTan2016equivocations} to the quantum
privacy amplification scenario. These results are presented in the following
two theorems. Theorem~\ref{theorem:3} deals with the asymptotic equivocation
rate, and Theorem~\ref{theorem:4} treats the security exponent.
In what follows, we use $|x|^+$ to denote $\max\{x, 0\}$.
\begin{theorem}
\label{theorem:3}
Let $\rho_{XE}$ be a CQ state, and $\mc{F}_{n}(R)$ be the set of functions
from $\mc{X}^n$ to $\mc{Z}_n=\{1,\ldots,2^{nR} \}$. Let $\rho^{f_n}_{Z_nE^n}$
denote the state resulting from applying a hash function $f_n\in\mc{F}
_{n}(R)$ to $\rho_{XE}^{\ox n}$. For any randomness extraction rate $R \geq 0$
and any $s \in (0,1]$, we have
\begin{equation}
\label{eq-easy}
\lim_{n\rightarrow\infty}
\frac{1}{n} \min_{f_n \in \mc{F}_n(R)} D_{1+s}(\rho^{f_n}_{Z_nE^n} \|
\frac{\1_{Z_n}}{|\mc{Z}_n|} \ox \rho_E^{\ox n})
={\big|R-H_{1+s}(X|E)_\rho \big|}^+.
\end{equation}
\end{theorem}

\begin{theorem}
\label{theorem:4}
Let $\rho_{XE}$ be a CQ state, and $\mc{F}_{n}(R)$ be the set of functions
from $\mc{X}^n$ to $\mc{Z}_n=\{1,\ldots,2^{nR} \}$. Let $\rho^{f_n}_{Z_nE^n}$
denote the state resulting from applying a hash function $f_n\in\mc{F}
_{n}(R)$ to $\rho_{XE}^{\ox n}$. For any randomness extraction rate $R \geq
R_{\rm critical}$ and any $s \in (0,1]$, we have
\begin{equation}
\label{eq:asyexpo}
\lim_{n\rightarrow\infty} -\frac{1}{n} \log \min_{f_n \in \mathcal{F}_n(R)}
D_{1+s}(\rho^{f_n}_{Z_nE^n} \| \frac{\1_{Z_n}}{|\mathcal{Z}_n|} \ox \rho_E^{\ox n})
= \big|\max_{t \in [s,1]} \big\{tH_{1+t}(X|E)_\rho -tR\big\}\big|^+.
\end{equation}
\end{theorem}

\begin{remark}
Actually, the results obtained in Section~\ref{sec:privacy-amp} already give
that for any randomness extraction rate $R \geq R_{\rm critical}$ and any
$s\in[-\frac{1}{2}, 0]$,
\begin{equation}
  \label{eq:remark3}
\lim_{n\rightarrow\infty} -\frac{1}{n} \log \min_{f_n \in \mathcal{F}_n(R)}
D_{1+s}(\rho^{f_n}_{Z_nE^n} \| \frac{\1_{Z_n}}{|\mathcal{Z}_n|} \ox \rho_E^{\ox n})
= \max_{0\leq t\leq 1} \big\{t\big(H_{1+t}(X|E)_{\rho}-R\big)\big\}.
\end{equation}
To see this, we first notice that $D_{\frac{1}{2}}(\rho\|\sigma)=-\log(1-P^2(\rho,\sigma))$.
This together with Eq.~\eqref{eq:exp-privacy-P} proves Eq.~\eqref{eq:remark3} for
one of the endpoint $s=-\frac{1}{2}$. On the other hand, Eq.~\eqref{eq:exp-privacy-D}
confirms Eq.~\eqref{eq:remark3} for the other endpoint $s=0$. As the function
$s\mapsto D_{1+s}(\rho\|\sigma)$ is monotonically increasing in $[-\frac{1}{2},0]$,
Eq.~\eqref{eq:remark3} for the whole interval follows.
\end{remark}

Before the proof of Theorem~\ref{theorem:3} and Theorem~\ref{theorem:4}, we
first present and prove several useful lemmas.
\begin{lemma}
\label{lemma:use1}
Let $A_x \in \mathcal{P}(\mathcal{H})$ for $x\in\mc{X}$, and let $\lambda$ be
a positive number. Then we have
\[\tr(\sum_{x \in \mc{X}} A_x-\lambda \1)_+ \geq \sum_{x \in \mc{X}} \tr(A_x-\lambda \1)_+.\]
\end{lemma}

\begin{proof}
Let $P_x=\{A_x>\lambda \1\}$ and let $P=\lor_x P_x$ be the projection onto the subspace spanned by $\{\supp(P_x)\}_{x \in \mc{X}}$. Then we have
\begin{equation}
\label{eq:lemma1}
\begin{split}
\tr(\sum_{x \in \mc{X}} A_x-\lambda \1)_+ &\geq \tr(\sum_{x \in \mc{X}} A_x-\lambda \1)P \\
&=\sum_x \tr A_xP- \lambda \tr P \\
& \geq \sum_x \tr A_x P_x- \lambda \tr P\\
&=\sum_x \tr(A_x-\lambda \1)_+ +\sum_x \lambda \tr P_x-\lambda \tr P \\
& \geq \sum_x \tr(A_x-\lambda \1)_+,
\end{split}
\end{equation}
where the first inequality is due to~(\ref{eq:pos}), the second inequality results from
$P \geq P_x, \forall x \in \mc{X}$, and the last inequality is because the sum of the
dimensions of all $P_x$ is larger than the dimension of $P$.
\end{proof}

\begin{lemma}
\label{lemma:use2}
Let $\rho_{XE}$, $\mc{Z}_n$, $\mathcal{F}_n(R)$ and $\rho^{f_n}_{Z_nE^n}$ be the same as
those in Theorem~\ref{theorem:3} and Theorem~\ref{theorem:4}. Then for any
$t>0$ and $0<R<H(X|E)_\rho$, we have
\[
\liminf_{n\rar\infty} \frac{1}{n}
\log \min_{f_n \in \mathcal{F}(R)} \tr (\rho^{f_n}_{Z_nE^n}-t\frac{\1_{Z_n}}{|\mc{Z}_n|}
\ox {\rho_E}^{\ox n})_+
\geq \inf_{s \geq 0} \big\{s(R-H_{1+s}(X|E)_{\rho})\big\}.
\]
\end{lemma}

\begin{proof}
Fix $m \in \mathbb{N}$, and write $n$ in the form $n=km+r$, where $k,r \in \mathbb{N}$
and $0 \leq r < m$. Suppose ${\rho^{\ox m}_E}$ and $\rho^{\ox r}_E$ have
spectral projections $\{E_i\}_{i \in \mathcal{I}}$ and $\{P_j\}_{j \in \mathcal{J}}$
with corresponding eigenvalues $\{\lambda_i\}_{i \in \mathcal{I}}$ and
$\{\eta_j\}_{j \in \mathcal{J}}$, respectively.

Now we evaluate the left hand side. First, recalling that the trace distance
decreases under the action of a channel, we have
\begin{equation}
\label{eq:lemma2}
\begin{split}
&\tr (\rho^{f_n}_{Z_nE^n}-t\frac{\1_{Z_n}}{|\mc{Z}_n|} \ox {\rho^{\ox n}_E})_+\\
=&\sum_{z_n}\tr(\sum_{x_n \in f_n^{-1}(z_n)}p(x_n)\rho^{x_n}_{E^n}-t\frac{{\rho^{\ox n}_E}}{|\mc{Z}_n|} )_+ \\
\geq & \sum_{z_n}\tr\big(\mathcal{E}_{\rho^{\ox m}_E}^{\ox k} \ox \mathcal{E}_{\rho^{\ox r}_E}(\sum_{x_n \in f_n^{-1}(z_n)}   {p(x_n)\rho^{x_n}_{E^n}}-t\frac{{\rho^{\ox n}_E}}{|\mc{Z}_n|}) \big)_+ \\
= & \sum_{z_n} \tr  \big(\sum_{i_k,j} (\sum_{x_n \in f_n^{-1}(z_n)}E_{i_k} \ox P_j p(x_n) \rho^{x_n}_{E^n} E_{i_k} \ox P_j
-t \frac{\lambda_{i_k} \eta_j E_{i_k} \ox P_j}{|\mc{Z}_n|})\big)_+ \\
=&\sum_{z_n} \sum_{i_k,j} \tr(\sum_{x_n \in f_n^{-1}(z_n)}E_{i_k} \ox P_j p(x_n) \rho^{x_n}_{E^n} E_{i_k} \ox P_j-t \frac{\lambda_{i_k} \eta_j E_{i_k} \ox P_j}{|\mc{Z}_n|})_+.
\end{split}
\end{equation}
Then, with Lemma~\ref{lemma:use1}, we can proceed as
\begin{equation}
\label{eq:lemma2-1}
\begin{split}
(\ref{eq:lemma2})& \geq \sum_{z_n} \sum_{i_k,j} \sum_{x_n \in f_n^{-1}(z_n)} \tr  (E_{i_k} \ox P_j p(x_n) \rho^{x_n}_{E^n} E_{i_k} \ox P_j
-t \frac{\lambda_{i_k} \eta_j E_{i_k} \ox P_j}{|\mc{Z}_n|})_+ \\
&= \sum_{z_n}\sum_{x_n \in f_n^{-1}(z_n)} \tr  \big(\sum_{i_k,j} E_{i_k} \ox P_j p(x_n) \rho^{x_n}_{E^n} E_{i_k} \ox P_j
-t \frac{\lambda_{i_k} \eta_j E_{i_k} \ox P_j}{|\mc{Z}_n|}\big)_+ \\
&= \sum_{x_n} \tr\big( \mathcal{E}_{\rho^{\ox m}_E}^{\ox k} \ox \mathcal{E}_{\rho^{\ox r}_E}(p(x_n)\rho^{x_n}_{E^n})- t\frac{{\rho^{\ox n}_E}}{|\mc{Z}_n|}  \big)_+ \\
&=\tr\big(\mathcal{E}_{\rho^{\ox m}_E}^{\ox k} \ox \mathcal{E}_{\rho^{\ox r}_E}(\rho^{\ox n}_{XE}) -t\frac{1}{2^{nR}}{\1^{\ox n}_X} \ox \rho^{\ox n}_E\big)_+  \\
& \geq  \tr \big( \mathcal{E}_{\rho^{\ox m}_E}^{\ox k} (\rho^{\ox mk}_{XE})-\frac{t|\mc{X}|^r}{2^{Rr}} \frac{\1^{\ox mk}_X \ox \rho^{\ox mk}_E}{2^{mkR}}   \big)_+,
\end{split}
\end{equation}
where the last inequality is because the trace distance decreases under partial trace.

Since the function $A \in \mathcal{P}(\mathcal{H}) \rightarrow \tr(A)^{1+s}$ is operator monotone,
Eq.~(\ref{eq:pinchingineq}) implies
\begin{align}
&v(\rho^{\ox m}_E)^{1+s}
2^{s D_{1+s}(\mathcal{E}_{\rho^{\ox m}_E}(\rho^{\ox m}_{XE})\| \1^{\ox m}_X \ox \rho^{\ox m}_E)}
=
v(\rho^{\ox m}_E)^{1+s}
Q_{1+s}(\mathcal{E}_{\rho^{\ox m}_E}(\rho^{\ox m}_{XE})\| \1^{\ox m}_X \ox \rho^{\ox m}_E)
\nonumber\\
\ge &
Q_{1+s}(\rho^{\ox m}_{XE} \| \1^{\ox m}_X \ox \rho^{\ox m}_E)
=
2^{sD_{1+s}(\rho^{\ox m}_{XE} \| \1^{\ox m}_X \ox \rho^{\ox m}_E)}.
\label{eq:lemma3H}
\end{align}
Therefore, we obtain
\begin{equation}
\label{eq:lemma3}
\begin{split}
&\liminf_{n\rar\infty} \frac{1}{n} \log \min_{f_n} \tr \big(\rho^{f_n}_{Z_nE^n}-t\frac{\1_{Z_n}}{|\mc{Z}_n|} \ox {\rho_E}^{\ox n}\big)_+ \\
\geq & \inf_{s \geq 0} s\big(R+\frac{D_{1+s}(\mathcal{E}_{\rho^{\ox m}_E}(\rho^{\ox m}_{XE}) \| \1^{\ox m}_X \ox \rho^{\ox m}_E)}{m}\big) \\
 \geq &\inf_{s \geq 0} s\big(R+\frac{D_{1+s}(\rho^{\ox m}_{XE} \| \1^{\ox m}_X \ox \rho^{\ox m}_E)}{m}
-\frac{(s+1)\log v(\rho^{\ox m}_E)}{sm}\big) \\
=&-\sup_{s \geq 0} \big\{s\big(H_{1+s}(X|E)_\rho-R+\frac{\log v(\rho^{\ox m}_E)}{m}\big)\big\}
  -\frac{\log v(\rho^{\ox m}_E)}{m}, \\
\end{split}
\end{equation}
where the first inequality follows from Eq.~(\ref{eq:lemma2-1}) and Proposition~\ref{prop:MosonyiOgawa},
and
the second inequality follows from Eq.~(\ref{eq:lemma3H}).

Because the function $R \mapsto \sup_{s \geq 0} {s(H_{1+s}(X|E)_\rho+R)}$
is continuous, by letting $m\rar\infty$ we conclude the proof.
\end{proof}

\begin{lemma}
\label{lemma:use3}
For a CQ state $\rho_{XE}$ and a two-universal random hash functions $F: \mc{X} \rightarrow \mc{Z}=
\{1, \ldots ,M\}$, we have for $s \in (0,1]$,
\begin{equation}
\mathbb{E}_{F} Q_{1+s}(\rho^{F}_{ZE} \| \1_Z \ox \rho_E) \leq v(\rho_E)^{1+s}
\big(Q_{1+s}(\rho_{XE} \| \1_X \ox \rho_E)+\frac{1}{M^s}\big).\label{lemma:use3EQ}
\end{equation}
\end{lemma}

\begin{proof}
Let the spectral projections of $\rho_E$ be $\{E_i\}_{i \in \mc{I}}$,
and the corresponding eigenvalues be $\{\lambda_i\}_{i \in \mc{I}}$.
Then, with the pinching inequality~(\ref{eq:pinchingineq}), we can bound
$\mathbb{E}_{F} Q_{1+s}(\rho^{F}_{ZE} \| \1_Z \ox \rho_E)$ as follows.
\begin{equation}
\label{eq:lemmause1}
\begin{split}
&\mathbb{E}_{F} Q_{1+s}(\rho^{F}_{ZE} \| {\1}_Z \ox \rho_E) \\
 \leq & v(\rho_E)^{1+s}\mathbb{E}_{F} Q_{1+s}\big({\mc{E}_{\rho_E}}(\rho^{F}_{ZE})
  \| {\1}_{Z} \ox \rho_E\big)\\
=&v(\rho_E)^{1+s}\mathbb{E}_{F}\big(\sum_{z,i}  Q_{1+s}(\sum_{x \in f^{-1}(z)}
  {\Pi}_i p(x)\rho^{x}_E {\Pi}_i \| \lambda_i {\Pi}_i) \big) \\
=&v(\rho_E)^{1+s}
\mathbb{E}_{F}\big(\sum_{z,i} \sum_{x \in f^{-1}(z)} {\lambda_i}^{-s}\tr
  {\Pi}_i p(x)\rho^{x}_E {\Pi}_i\big({\Pi}_i p(x)\rho^{x}_E {\Pi}_i +\sum_{x
  \neq x'} {\Pi}_i p(x')\rho^{x'}_E {\Pi}_i 1_{f(x')
  \neq f(x)}\big)^s\big),
\end{split}
\end{equation}
where the inequality follows from the same reason as Eq.~(\ref{eq:lemma3H}).
To proceed, we invoke the property that the function $f(x)=x^s$ is operator
concave when $0 < s \leq 1$, to see that
\begin{equation}
\label{eq:lemmause1-1}
\begin{split}
&\mathbb{E}_{F}\big(\sum_{z,i} \sum_{x \in f^{-1}(z)} {\lambda_i}^{-s}\tr
  {\Pi}_i p(x)\rho^{x}_E {\Pi}_i\big({\Pi}_i p(x)\rho^{x}_E {\Pi}_i +\sum_{x
  \neq x'} {\Pi}_i p(x')\rho^{x'}_E {\Pi}_i 1_{f(x')
  \neq f(x)}\big)^s\big)\\
\leq & \sum_{z,i} \sum_{x \in f^{-1}(z)} {\lambda_i}^{-s} \tr {\Pi}_i p(x)
  \rho^{x}_E {\Pi}_i\big({\Pi}_i p(x)\rho^{x}_E {\Pi}_i+\sum_{x \neq x'}
  \frac{1}{M} {\Pi}_i p(x')\rho^{x'}_E {\Pi}_i\big)^s\\
=& \sum_{z,i} \sum_{x \in f^{-1}(z)} {\lambda_i}^{-s}\tr {\Pi}_i p(x)\rho^{x}_E
  {\Pi}_i\big(\frac{1}{M}{\Pi}_i \rho_E {\Pi}_i+\frac{M-1}{M} {\Pi}_i p(x)\rho^x_E
  {\Pi}_i  \big)^s.
\end{split}
\end{equation}
Then we use the inequality $(X+\lambda\1)^s \leq X^s + \lambda^s \1$ for any $X \in \mathcal{P}(\mathcal{H})$ and $\lambda \geq 0$, to bound Eq.~(\ref{eq:lemmause1-1}) as follows.
\begin{equation}
\label{eq:lemmause1-2}
\begin{split}
&\sum_{z,i} \sum_{x \in f^{-1}(z)} {\lambda_i}^{-s}\tr {\Pi}_i p(x)\rho^{x}_E
  {\Pi}_i\big(\frac{1}{M}{\Pi}_i \rho_E {\Pi}_i+\frac{M-1}{M} {\Pi}_i p(x)\rho^x_E
  {\Pi}_i  \big)^s \\
\leq  & \sum_i \sum_x  \lambda^{-s}_{i}\tr {\Pi}_i p(x)\rho^{x}_E {\Pi}_i\big(
  \frac{1}{M^s} ({\Pi}_i \rho_E {\Pi}_i)^s +\frac{(M-1)^s}{M^s} \big({\Pi}_i p(x)
  \rho^{x}_E {\Pi}_i\big)^s  \big)      \\
 \leq &  \sum_x \sum_i \tr(\lambda^{-s}_i({\Pi}_i p(x)\rho^x_E
  {\Pi}_i)^{1+s}+\frac{1}{M^s}{\Pi}_i p(x) \rho^x_E {\Pi}_i    )\\
  = & \big( Q_{1+s}(\mc{E}_{\rho_E}(\rho_{XE}) \| \1_X \ox \rho_E)
    +\frac{1}{M^s}  \\
 \leq & Q_{1+s}(\rho_{XE} \| \1_X \ox \rho_E)+\frac{1}{M^s} ,
\end{split}
\end{equation}
where the second inequality is simply due to $\frac{M-1}{M} <1$.
Therefore, the combination of Eqs.~(\ref{eq:lemmause1}), (\ref{eq:lemmause1-1}),
and (\ref{eq:lemmause1-2})
yields Eq.~\eqref{lemma:use3EQ}.
\end{proof}

\begin{lemma}
\label{lemma:use4}
For a CQ state $\rho_{XE}$ and a two-universal random hash functions $F: \mc{X} \rightarrow \mc{Z}=
\{1, \ldots ,M\}$, we have for $s \in (0,1]$,
\begin{equation}
\mathbb{E}_{F} 2^{sD_{1+s}(\rho^{F}_{ZE} \| \frac{\1_Z}{|\mc{Z}|} \ox \rho_E)} \leq
1+v(\rho_E)^s 2^{s(\log M -H_{1+s}(X|E)_\rho)}.
\end{equation}
\end{lemma}

\begin{proof}
For $s \in (0,1]$, we have
\begin{equation}
\label{equ:use4}
\begin{split}
&\mathbb{E}_{F} 2^{sD_{1+s}(\rho^{F}_{ZE} \| \1_Z \ox \rho_E)} \\
=&\tr \mathbb{E}_{F} \sum_{m \in \mc{Z}}
\big(\rho_{E}^{-\frac{s}{2(1+s)}} \big(\sum_{x': F(x')=m} p(x')\rho_{E}^{x'} \big)   \rho_{E}^{-\frac{s}{2(1+s)}}\big)^{1+s}   \\
=&\tr \mathbb{E}_{F} \sum_{m \in \mc{Z}}
\big(\rho_{E}^{-\frac{s}{2(1+s)}} \big(\sum_{x': F(x')=m} p(x')\rho_{E}^{x'} \big)   \rho_{E}^{-\frac{s}{2(1+s)}}\big) \big(\rho_{E}^{-\frac{s}{2(1+s)}} \big(\sum_{x': F(x')=m} p(x')\rho_{E}^{x'} \big)   \rho_{E}^{-\frac{s}{2(1+s)}}\big)^s \\
=&\tr \mathbb{E}_{F} \sum_{m \in \mc{Z}}
\big(\sum_{x': F(x')=m} p(x')\rho_{E}^{x'} \big)   \rho_{E}^{-\frac{s}{2(1+s)}} \big(\rho_{E}^{-\frac{s}{2(1+s)}} \big(\sum_{x': F(x')=m} p(x')\rho_{E}^{x'} \big)   \rho_{E}^{-\frac{s}{2(1+s)}}\big)^s\rho_{E}^{-\frac{s}{2(1+s)}} \\
=&\tr \mathbb{E}_{F} \sum_{x \in \mathcal{X}}
p(x)\rho_E^x  \rho_{E}^{-\frac{s}{2(1+s)}} \big(\rho_{E}^{-\frac{s}{2(1+s)}} \big(\sum_{x': F(x')=F(x)} p(x')\rho_{E}^{x'} \big)   \rho_{E}^{-\frac{s}{2(1+s)}}\big)^s\rho_{E}^{-\frac{s}{2(1+s)}}.
\end{split}
\end{equation}
Then we proceed as follows.
\begin{equation}
\begin{split}
&\mathbb{E}_{F} 2^{sD_{1+s}(\rho^{F}_{ZE} \| \1_Z \ox \rho_E)} \\
\stackrel{(a)}{\leq} &\tr \sum_{x \in \mathcal{X}}
p(x)\rho_E^x  \rho_{E}^{-\frac{s}{2(1+s)}} \big(\rho_{E}^{-\frac{s}{2(1+s)}} \big(\mathbb{E}_{F}\sum_{x': F(x')=F(x)} p(x')\rho_{E}^{x'} \big)   \rho_{E}^{-\frac{s}{2(1+s)}}\big)^s\rho_{E}^{-\frac{s}{2(1+s)}} \\
\stackrel{(b)}{\leq} &\tr \sum_{x \in \mathcal{X}}
p(x)\rho_E^x  \rho_{E}^{-\frac{s}{2(1+s)}} \big(\rho_{E}^{-\frac{s}{2(1+s)}} \big(v(\rho_E) \mc{E}_{\rho_E}
(p(x)\rho_E^x)+\frac{1}{M}\rho_E\big)   \rho_{E}^{-\frac{s}{2(1+s)}}\big)^s\rho_{E}^{-\frac{s}{2(1+s)}} \\
=&\sum_{x \in \mathcal{X}} \tr p(x)\rho_E^x \rho_E^{-s}\big(v(\rho_E) \mc{E}_{\rho_E}
(p(x)\rho_E^x)+\frac{1}{M}\rho_E  \big)^s \\
\leq & \sum_{x \in \mathcal{X}} \tr p(x)\rho_E^x \rho_E^{-s}\big(v(\rho_E)^s (\mc{E}_{\rho_E}
(p(x)\rho_E^x))^s+\frac{1}{M^s}\rho^s_E  \big)\\
=&v(\rho_E)^s 2^{-sH_{1+s}(X|E)_{\mathcal{E}_{\rho_E}(\rho)}}+\frac{1}{M^s}  \\
\stackrel{(c)}{\leq} & v(\rho_E)^s 2^{-sH_{1+s}(X|E)_\rho}+\frac{1}{M^s},
\end{split}
\end{equation}
where $(a)$ follows from the matrix concavity of $x \mapsto x^s$, $(b)$ comes from
the definition of the two-universal hash functions and the pinching inequality~(\ref{eq:pinchingineq})
and $(c)$ is due to the data processing inequality of the sandwiched R\'enyi divergence. This completes
the proof.
\end{proof}

\medskip
Now, we are ready for the proofs of Theorem~\ref{theorem:3} and Theorem~\ref{theorem:4}.

\begin{proofof}{Theorem~\ref{theorem:3}}
At first, we deal with the "$\leq$" part. By Lemma~\ref{lemma:use3},
we know that for any $n$ there exists a hash function $f_n$ such that
\begin{equation}
\begin{split}
& D_{1+s}(\rho^{f_n}_{Z_nE^n}\|\frac{\1_{Z_n}}{|\mc{Z}_n|}\ox\rho_E^{\ox n}) \\
=&nR+\frac{1}{s}\log Q_{1+s}(\rho^{f_n}_{Z_nE^n} \| \1_{Z_n} \ox \rho^{\ox n}_E) \\
\leq& nR+ \frac{1}{s}\log \big(Q^n_{1+s}(\rho_{XE} \| \1_{X} \ox \rho_E)
  +\frac{1}{2^{nRs}} \big)+\frac{1+s}{s}\log v(\rho_E^{\ox n}) \\
=&\frac{1}{s}\log \big(1+2^{nRs} Q^n_{1+s}(\rho_{XE} \| \1_X \ox \rho_E)\big)
  +\frac{1+s}{s}\log v(\rho_E^{\ox n}).
\end{split}
\end{equation}
This further yields
\[
\limsup_{n\rar\infty} \frac{1}{n} \min_{f_n \in \mc{F}_n (R)}D_{1+s}(\rho^
{f_n}_{Z_nE^n} \| \frac{\1_{Z_n}}{|\mc{Z}_n|} \ox \rho_E^{\ox n})
\leq {\big|R-H_{1+s}(X|E)_\rho \big|}^+.
\]

Next, we turn to the derivation of the other direction. By
Proposition \ref{prop:srd-properties}~(\romannumeral6), we have for any hash function $f_n$,
\begin{equation}
\begin{split}
&D_{1+s}(\rho^{f_n}_{Z_nE^n} \| \frac{\1_{Z_n}}{|\mc{Z}_n|} \ox \rho_E^{\ox n})\\
=&nR-H_{1+s}(Z_n | E^n)_{\rho^{f_n}_{Z_nE^n} } \\
 \geq & nR-nH_{1+s}(X | E)_{\rho}.
\end{split}
\end{equation}
From this and noticing that $D_{1+s}(\rho^{f_n}_{Z_nE^n} \| \frac{\1_{Z_n}}
{|\mc{Z}_n|} \ox \rho_E^{\ox n}) \geq 0$, it is easy to get
\[
\liminf_{n\rar\infty} \frac{1}{n} \min_{f_n \in \mc{F}(R)}D_{1+s}
(\rho^{f_n}_{Z_nE^n} \| \frac{\1_{Z_n}}{|\mc{Z}_n|} \ox \rho_E^{\ox n})
\geq \big|R-H_{1+s}(X|E)_\rho \big|^+.
\]
\end{proofof}

\medskip
\begin{proofof}{Theorem~\ref{theorem:4}}
At first, we prove the "$\geq$" part.
The left side of Eq.~(\ref{eq:asyexpo}) can be bounded as follows, thanks to the monotonicity of the sandwiched R{\'e}nyi divergence~\big(Proposition \ref{prop:srd-properties}~(\romannumeral1)\big) and Lemma \ref{lemma:use4}.
\begin{equation}
\label{eq:theorem17}
\begin{split}
&\min_{f_n \in \mathcal{F}_n(R)} D_{1+s}(\rho^{f_n}_{Z_nE^n} \| \frac{\1_{Z_n}}{|\mathcal{Z}_n|} \ox \rho_E) \\
\leq & \min_{f_n \in \mathcal{F}_n(R)} D_{1+t}(\rho^{f_n}_{Z_nE^n} \| \frac{\1_{Z_n}}{|\mathcal{Z}_n|} \ox \rho_E) \\
\leq & \frac{1}{t} \log\big(1+v(\rho^{\ox n}_E)^t 2^{t(nR -nH_{1+t}(X|E)_\rho)}\big)\\
\leq& \frac{\log e}{t} v(\rho^{\ox n}_E)^t 2^{t(nR -nH_{1+t}(X|E)_\rho)},
\end{split}
\end{equation}
for any $t \in [s,1]$. Noticing that $D_{1+s}(\rho^{f_n}_{Z_nE^n} \| \frac{\1_{Z_n}}{|\mathcal{Z}_n|} \ox \rho_E) \leq nR$, we see that the exponent must be non-negative. This
observation and Eq. (\ref{eq:theorem17}) implies
\[
\liminf_{n\rar\infty}
-\frac{1}{n} \log \min_{f_n \in \mathcal{F}_n(R)}
D_{1+s}(\rho^{f_n}_{Z_nE^n} \| \frac{\1_{Z_n}}{|\mathcal{Z}_n|} \ox \rho_E)
\geq \big|\sup_{t \in [s,1]} tH_{1+t}(X|E)_\rho -tR\big|^+.
\]

Next, we prove the other direction. We will deal with the case $R \leq \hat{R}(s)$
and the case $R > \hat{R}(s)$ separately. Now we start with the former case. Let
$f_n$ be an arbitrary hash function. We choose a positive constant $c$ such that
$c^s-2 \geq 1$. Then we construct a channel
\[
\Phi(X)=(\tr \{\rho^{f_n}_{Z_nE^n}>c \frac{\1_{Z_n}}{|\mc{Z}_n|} \ox \rho^{\ox n}_E  \} X) \proj{0}
+(\tr\{\rho^{f_n}_{Z_nE^n} \leq c \frac{\1_{Z_n}}{|\mc{Z}_n|} \ox \rho^{\ox n}_E  \} X) \proj{1},
\]
and denote
\[
  p_n=\tr\rho^f_{Z_nE^n}\{\rho^{f_n}_{Z_nE^n}>c \frac{\1_{Z_n}}{|\mc{Z}_n|} \ox \rho^{\ox n}_E  \} \quad\text{and}\quad
 q_n=\tr(\frac{\1_{Z_n}}{|\mc{Z}_n|} \ox \rho^{\ox n}_E)\{\rho^{f_n}_{Z_nE^n}>c \frac{\1_{Z_n}}{|\mc{Z}_n|} \ox \rho^{\ox n}_E\} .
\]
It is easy to see that
\begin{equation}
\label{basic}
p_n \geq c q_n.
\end{equation}
Hence, by the data processing inequality for the channel $\Phi$ and Eq.~(\ref{basic}), we have
\begin{equation}
\label{equation:1}
\begin{split}
&D_{1+s}(\rho^{f_n}_{Z_nE^n} \| \frac{\1_{Z_n}}{|\mc{Z}_n|}\ox \rho^{\ox n}_E) \\
\geq &\frac{1}{s}\log \big\{p^{s+1}_n q^{-s}_n+(1-p_n)^{s+1} (1-q_n)^{-s}\big\} \\
\geq&\frac{1}{s}\log \big\{c^s p_n+(1-p_n)^2 \big\} \\
=&\frac{1}{s}\log \big\{ 1+(c^s-2)p_n+p^2_n\big\}\\
\geq &\frac{1}{s}\log \big\{1+p_n\big\} \\
\geq &\frac{1}{s}\log \big\{1+\tr (\rho^{f_n}_{Z_nE^n}-c\frac{\1_{Z_n}}{|\mc{Z}_n|}\ox \rho^{\ox n}_E)_+\big\},
\end{split}
\end{equation}
where the third inequality follows from $c^s-2>1$.

Eq.~(\ref{equation:1}) implies
\begin{equation}
\label{euive}
\begin{split}
\min_{f_n \in \mathcal{F}_n(R)}
D_{1+s}(\rho^{f_n}_{Z_nE^n} \| \frac{\1_{Z_n}}{|\mc{Z}_n|}\ox \rho^{\ox n}_E)& \geq
 \frac{1}{s}\log \big\{1+\min_{f_n \in \mathcal{F}_n(R)}\tr (\rho^{f_n}_{Z_nE^n}-c\frac{\1_{Z_n}}{|\mc{Z}_n|}\ox \rho^{\ox n}_E)_+\big\} \\
&\overset{.}{=} \frac{1}{s}\min_{f_n \in \mathcal{F}_n(R)}\tr (\rho^{f_n}_{Z_nE^n}-c\frac{\1_{Z_n}}{|\mc{Z}_n|}\ox \rho^{\ox n}_E)_+,
\end{split}
\end{equation}
where $a_n \overset{.}{=} b_n$ means that
$\lim_{n\rightarrow\infty} \frac{1}{n} \log \frac{a_n}{b_n} =0$.
Now, we can use Eq.~(\ref{euive}) and Lemma~\ref{lemma:use2} to get
\begin{equation}
\limsup_{n\rar\infty}
-\frac{1}{n} \log \min_{f_n \in \mathcal{F}_n(R)}
D_{1+s}(\rho^{f_n}_{Z_nE^n} \| \frac{\1_{Z_n}}{|\mc{Z}_n|}\ox \rho^{\ox n}_E)
\leq \sup_{t \geq 0} (tH_{1+t}(X|E)_\rho -tR).
\end{equation}
Because $\hat{R}(1) \leq R \leq \hat{R}(s)$, we have
 \[
 \sup_{t \geq 0} (tH_{1+t}(X|E)_\rho -tR)=
 \max_{s \leq t \leq 1} (tH_{1+t}(X|E)_\rho -tR),
 \]
and we complete the case $R \leq \hat{R}(s)$.

Next, we turn to the case $R > \hat{R}(s)$. We also define a channel like the above step
\[
\Delta(X)=(\tr \{\rho^{f_n}_{Z_nE^n}> \frac{\1_{Z_n}}{2^{n\hat{R}(s)}} \ox \rho^{\ox n}_E  \} X) \proj{0}
+(\tr\{\rho^{f_n}_{Z_nE^n} \leq  \frac{\1_{Z_n}}{2^{n\hat{R}(s)}} \ox \rho^{\ox n}_E  \} X) \proj{1},
\]
and denote
\[
p_n=\tr\rho^{f_n}_{Z_nE^n}\{\rho^{f_n}_{Z_nE^n}> \frac{\1_{Z_n}}{2^{n\hat{R}(s)}} \ox \rho^{\ox n}_E  \} \quad\text{and}\quad
 q_n=\tr(\frac{\1_{Z_n}}{|\mc{Z}_n|}\ox \rho^{\ox n}_E)\{\rho^{f_n}_{Z_nE^n}> \frac{\1_{Z_n}}{2^{n\hat{R}(s)}} \ox \rho^{\ox n}_E\} .
\]
We invoke a similar argument as Eq. (\ref{equation:1}) to bound $D_{1+s}(\rho^{f_n}_{Z_nE^n} \| \frac{\1_{Z_n}}{|\mc{Z}_n|}\ox \rho^{\ox n}_E)$, by using the channel $\Delta$.
\begin{equation}
\label{equation:2}
\begin{split}
&D_{1+s}(\rho^{f_n}_{Z_nE^n} \| \frac{\1_{Z_n}}{|\mc{Z}_n|}\ox \rho^{\ox n}_E)\\
\geq &\frac{1}{s}\log \big\{p^{s+1}_n q^{-s}_n+(1-p_n)^{s+1} (1-q_n)^{-s}\big\} \\
\geq&\frac{1}{s}\log \big\{2^{ns(R-\hat{R}(s))} p_n+(1-p_n)^2\big\}\\
\geq &\frac{1}{s}\log \big\{1+(2^{sn(R-\hat{R}(s))}-2)p_n\big\} \\
\geq &\frac{1}{s}\log \big\{1+(2^{sn(R-\hat{R}(s))}-2)(\rho^f_{Z_nE^n}-c\frac{\1_{Z_n}}{2^{n\hat{R}(s)}}\ox \rho^{\ox n}_E)_+\big\}.
\end{split}
\end{equation}
Eq.~(\ref{equation:2}) and Lemma \ref{lemma:use2} imply
\begin{equation}
\label{equation:3}
\begin{split}
&\min_{f_n \in \mathcal{F}_n(R)}
D_{1+s}(\rho^{f_n}_{Z_nE^n} \| \frac{\1_{Z_n}}{|\mc{Z}_n|}\ox \rho^{\ox n}_E) \\
\overset{.}{\geq} &\frac{1}{s}\log \big\{1+ 2^{sn(R-\hat{R}(s))}2^{n\,{\inf_{t \geq 0} (t\hat{R}(s)-tH_{1+t}(X|E)_\rho )}} \big\} \\
=&\frac{1}{s}\log \big\{1+2^{sn(R-\hat{R}(s))}2^{n(s\hat{R}(s)-sH_{1+s}(X|E)_\rho )}\big\},
\end{split}
\end{equation}
where $a_n \overset{.}{\geq} b_n$ means that
$\lim\limits_{n\rightarrow\infty} \frac{\log a_n}{n} \geq \lim\limits_{n\rightarrow\infty} \frac{\log b_n}{n}$,
and the last line is because the minimum of the function $t \mapsto
t(\hat{R}(s)-H_{1+t}(X|E)_\rho)$ is achieved at $t=s$ when $R > \hat{R}(s)$.

Eq. (\ref{equation:3}) further gives
\begin{equation}
\begin{split}
\label{equation:4}
\limsup_{n\rar\infty}
-\frac{1}{n} \log \min_{f_n \in \mathcal{F}_n(R)}
D_{1+s}(\rho^{f_n}_{Z_nE^n} \| \frac{\1_{Z_n}}{|\mc{Z}_n|}\ox \rho^{\ox n}_E)
&\leq \big|sH_{1+s}(X|E)_\rho-sR\big|_+ \\
&=\big|\max_{s \leq t \leq 1} (tH_{1+t}(X|E)_\rho -tR)\big|_+,
\end{split}
\end{equation}
and this completes the proof of the case $R > \hat{R}(s)$.
\end{proofof}

\section{Conclusion and discussion}
  \label{sec:discussion}
Employing the sandwiched R{\'e}nyi divergence, we have obtained the precise exponent
in smoothing the max-relative entropy, and as an application,
combining the existing result~\cite[Theorem 1]{Hayashi2015precise},
we have also obtained
the precise exponent for quantum privacy amplification when the rate of extracted
randomness is not too low. Our results, along with the concurrent
work~\cite{LiYao2021reliability} which addresses different problems,
clearly show that the sandwiched R{\'e}nyi
divergence can not only characterize the strong converse exponents~\cite{
MosonyiOgawa2015quantum, MosonyiOgawa2015two,CMW2016strong,
HayashiTomamichel2016correlation, MosonyiOgawa2017strong, CHDH2020non}, but
also accurately characterizes how the performance of certain
quantum information processing tasks approach the perfect. We anticipate that more
applications of the sandwiched R{\'e}nyi divergence along this line will be found
in the future.

Different definitions for the sandwiched R{\'e}nyi conditional entropy have been proposed,
among which two typical versions are~\cite{TBH2014relating, MDSFT2013on}
\begin{align}
H_{\alpha}(A|B)_\rho &=-D_{\alpha}(\rho_{AB} \| \1_A \ox \rho_B), \quad\text{and}
\label{eq:srce-a} \\
\bar{H}_{\alpha}(A|B)_\rho&=-\min_{\sigma_B\in\mc{S}(\mc{H}_B)}D_{\alpha}(\rho_{AB}
\| \1_A \ox \sigma_B),
\label{eq:srce-b}
\end{align}
and it was not quite clear which one should be the proper formula. The
version~(\ref{eq:srce-b}) has later found operational meanings in Ref.~\cite{HayashiTomamichel2016correlation} and Ref.~\cite{CHDH2020non}.
By giving an operational meaning to the version~(\ref{eq:srce-a}) in this paper,
we conclude that both versions are proper expressions and the sandwiched
R{\'e}nyi conditional entropy is not unique.

The smoothing quantity in Theorem~\ref{theorem:exp-mre} and the insecurity in
Theorem~\ref{theorem:exp-privacy-up} as well as Theorem~\ref{theorem:exp-privacy} are
measured by the purified distance and/or the Kullback-Leibler divergence.
Determining the respective exponents for these two problems under the trace distance
is an interesting open problem.
Originally, Renner~\cite{Renner2005security} defined the smoothing of the
max-relative entropy based on the trace norm distance to derive an upper bound
of the insecurity in privacy amplification under two-universal hashing. However,
the reference~\cite{Hayashi2016} showed that this type of entropy cannot derive
the tight exponential upper bound in the classical setting of this problem
while it derived the type exponential behavior based on the trace norm distance
in the classical case.
Instead, the references~\cite{Hayashi2013,Hayashi2016} showed that the smoothing
of the R\'{e}nyi entropy of order $2$ based on the trace norm distance derives
the tight exponential upper bound in the classical setting of this problem.
The reference~\cite{Hayashi2014} considered its quantum extension, but did not
derive the tight exponential evaluation, while this topic has a recent
progress~\cite{Dupuis2021privacy} after the references~\cite{Hayashi2013,
Hayashi2014, Hayashi2016}.

For privacy amplification, we are only able to find
out the exact exponent when the rate $R$ of the randomness extraction is above the
critical value $R_{\rm critical}$. Determining the exponent for rate $R$ less than
$R_{\rm critical}$ is another important open question. The examples in
Section~\ref{subsec:low-rate} indicate that this problem may be more of a
combinatorial feature in the low-rate regime, at least when the rate $R$ is such
that $0\leq R \leq H_\text{min}(X|E)_\rho$.

In addition, Section \ref{sec:equivocation} has derived
the asymptotic equivocation rate under the sandwiched R\'enyi divergence
for any randomness extraction rate as Theorem \ref{theorem:3}.
Also, this section has derived the security exponent under the sandwiched R\'enyi
divergence in Theorem \ref{theorem:4} when the randomness extraction rate is
not smaller than the critical rate. This exponent is remained an open problem when
the randomness extraction rate is larger than the critical rate $R_{\rm critical}$.

\section*{Acknowledgements}
We are grateful to the anonymous referees for the valuable suggestions, which have
helped us improve the manuscript. The research of KL was supported by the National
Natural Science Foundation of
China  (No. 61871156, No. 12031004). The research of YY was supported by the
National Natural Science Foundation of China  (No. 61871156, No. 12071099).
MH is supported in part by the National Natural Science Foundation of
China (No. 62171212) and Guangdong Provincial Key Laboratory (No. 2019B121203002).

{\appendix[Proof of Proposition~\ref{prop:monohash}]
We need the following lemma.
\begin{lemma}
\label{lemma:isometry}
Let $\sigma_{AB} \in \mc{S}(\mc{H}_{AB})$ and let
$U: \mc{H}_A \rightarrow \mc{H}_{A'}$ be an isometry. Then
\[
H^{\epsilon}_{\rm{min}}(A|B)_\sigma=H^{\epsilon}_{\rm{min}}(A'|B)_{U \sigma U^*}.
\]
\end{lemma}

\begin{proof}
By definition, there is a state $\tilde{\sigma}_{AB} \in \mc{B}^\epsilon(\sigma_{AB})$
satisfying
\[
  \tilde{\sigma}_{AB} \leq 2^{-H^{\epsilon}_{\text{min}}(A|B)_\sigma}\mathbbm{1}_A \ox \sigma_B.
\]
Let $\tilde{\sigma}_{A'B}:=U\tilde{\sigma}_{AB}U^*$. Obviously, we have
$\tilde{\sigma}_{A'B} \in \mc{B}^\epsilon(U\sigma_{AB}U^*)$, and
\[\begin{split}
  \tilde{\sigma}_{A'B} &\leq 2^{-H^{\epsilon}_{\text{min}}(A|B)_\sigma}U\mathbbm{1}_A U^* \ox \sigma_B \\
                       &\leq 2^{-H^{\epsilon}_{\text{min}}(A|B)_\sigma} \mathbbm{1}_{A'} \ox \sigma_B.
\end{split}\]
This verifies by definition that
\[
H^{\epsilon}_{\text{min}}(A|B)_\sigma \leq H^{\epsilon}_{\text{min}}(A'|B)_{U \sigma U^*}.
\]
For the opposite direction, similarly, by definition there is a state $\tilde{\sigma}_{A'B}
\in \mc{B}^\epsilon(U\sigma_{AB}U^*)$ satisfying
\[
  \tilde{\sigma}_{A'B} \leq 2^{-H^{\epsilon}_{\text{min}}(A'|B)_{U\sigma U^*}}\mathbbm{1}_{A'} \ox \sigma_B.
\]
Then for the subnormalized state $U^*\tilde{\sigma}_{A'B}U\in\mc{S}_\leq(\mc{H}_{AB})$,
we can check that
\begin{equation}\begin{split}
  \label{eq:iso-a}
P(\sigma_{AB}, U^*\tilde{\sigma}_{A'B}U) &= P(U\sigma_{AB}U^*, UU^*\tilde{\sigma}_{A'B}UU^*) \\
                                         &= P(U\sigma_{AB}U^*, \tilde{\sigma}_{A'B})         \\
                                         &\leq \epsilon,
\end{split}\end{equation}
and
\[\begin{split}
    U^*\tilde{\sigma}_{A'B}U
  &\leq 2^{-H^{\epsilon}_{\text{min}}(A'|B)_{U\sigma U^*}}U^*\mathbbm{1}_{A'}U\ox\sigma_B \\
  & =   2^{-H^{\epsilon}_{\text{min}}(A'|B)_{U\sigma U^*}} \mathbbm{1}_{A} \ox \sigma_B,
\end{split}\]
where for the second line of Eq.~(\ref{eq:iso-a}), notice that $UU^*$ is a projection
onto $U\mc{H}_A$, and hence we check it directly using the expression of the
fidelity function. This implies by definition that
\[
H^{\epsilon}_{\text{min}}(A|B)_\sigma \geq H^{\epsilon}_{\text{min}}(A'|B)_{U \sigma U^*}.
\]
\end{proof}

\begin{proofof}{Proposition~\ref{prop:monohash}}
Let $U:\ket{x}\mapsto\ket{x}\ox\ket{f(x)}$ be the isometry from $X$ to $XZ$, and write
$\sigma_{XZAB}=U\sigma_{XAB}U^*$. Obviously, $\sigma_{XZAB}$ is classical on $X$ and $Z$,
and is the extension of both $\sigma_{XAB}$ and $\sigma_{ZAB}$. Since Lemma~\ref{lemma:isometry}
gives that $H^{\epsilon}_{\text{min}}(XA|B)_\sigma=H^{\epsilon}_{\text{min}}(XZA|B)_\sigma$,
what we need to do is to show
\begin{equation}
  \label{eq:toshow}
  H^{\epsilon}_{\text{min}}(XZA|B)_\sigma \geq H^{\epsilon}_{\text{min}}(ZA|B)_\sigma.
\end{equation}
By the definition of $H^{\epsilon}_{\text{min}}(ZA|B)_\sigma$, there is
$\tilde{\sigma}_{ZAB}\in\mc{B}^\epsilon(\sigma_{ZAB})$ such that
\begin{equation}
  \label{eq:monohash-a}
\tilde{\sigma}_{ZAB}\leq 2^{-H^{\epsilon}_{\text{min}}(ZA|B)_\sigma} \mathbbm{1}_{ZA}\ox\sigma_B.
\end{equation}
Now Uhlmann's theorem~\cite{Uhlmann1976the} tells us that there is $\hat{\sigma}
_{XZAB}\in\mc{S}_\leq(\mc{H}_{XZAB})$ which extends $\tilde{\sigma}_{ZAB}$ and satisfies
$P(\sigma_{XZAB},\hat{\sigma}_{XZAB})=P(\sigma_{ZAB},\tilde{\sigma}_{ZAB})$. Using the
measurement map $\mc{M}_X:L\mapsto\sum_x\proj{x}L\proj{x}$, we define
$\tilde{\sigma}_{XZAB}:=\mc{M}_X(\hat{\sigma}_{XZAB})$. Since $\sigma_{XZAB}=
\mc{M}_X(\sigma_{XZAB})$,
\begin{equation}
  \label{eq:monohash-b}
  P(\sigma_{XZAB},\tilde{\sigma}_{XZAB}) \leq P(\sigma_{XZAB},\hat{\sigma}_{XZAB}) \leq \epsilon.
\end{equation}
By construction, $\tilde{\sigma}_{XZAB}$ has the form $\tilde{\sigma}_{XZAB}=\sum_x
\proj{x}_X\ox \tilde{\sigma}^x_{ZAB}$ and is still an extension of $\tilde{\sigma}_{ZAB}$.
So, $\tilde{\sigma}^x_{ZAB}\leq\sum_x\tilde{\sigma}^x_{ZAB}=\tilde{\sigma}_{ZAB}$. This,
together with Eq.~(\ref{eq:monohash-a}), ensures that
\begin{equation}
  \label{eq:monohash-c}
\tilde{\sigma}_{XZAB}\leq 2^{-H^{\epsilon}_{\text{min}}(ZA|B)_\sigma} \mathbbm{1}_{XZA}\ox\sigma_B.
\end{equation}
Eq.~(\ref{eq:monohash-b}) and Eq.~(\ref{eq:monohash-c}) together imply Eq.~(\ref{eq:toshow}),
concluding the proof.
\end{proofof} }


\begin{thebibliography}{99}

\bibitem{Renner2005security}
R.~Renner, ``Security of quantum key distribution,'' {\em Ph. D. Thesis}, 2005.

\bibitem{Datta2009min}
N.~Datta, ``Min-and max-relative entropies and a new entanglement monotone,''
  {\em IEEE Transactions on Information Theory}, vol.~55, no.~6,
  pp.~2816--2826, 2009.

\bibitem{TCR2009fully}
M.~Tomamichel, R.~Colbeck, and R.~Renner, ``A fully quantum asymptotic
  equipartition property,'' {\em IEEE Transactions on Information Theory},
  vol.~55, no.~12, pp.~5840--5847, 2009.

\bibitem{TCR2010duality}
M.~Tomamichel, R.~Colbeck, and R.~Renner, ``Duality between smooth min-and
  max-entropies,'' {\em IEEE Transactions on information theory}, vol.~56,
  no.~9, pp.~4674--4681, 2010.

\bibitem{BrandaoPlenio2010reversible}
F.~G. Brandao and M.~B. Plenio, ``A reversible theory of entanglement and its
  relation to the second law,'' {\em Communications in Mathematical Physics},
  vol.~295, no.~3, pp.~829--851, 2010.

\bibitem{BCR2011the}
M.~Berta, M.~Christandl, and R.~Renner, ``The quantum reverse {Shannon} theorem
  based on one-shot information theory,'' {\em Communications in Mathematical
  Physics}, vol.~306, no.~3, p.~579, 2011.

\bibitem{Tomamichel2015quantum}
M.~Tomamichel, {\em Quantum information processing with finite resources:
  mathematical foundations}, vol.~5.
\newblock Springer, 2015.

\bibitem{Han2000}
T.~S. Han, ``Hypothesis testing with the general source,'' {\em IEEE
  Transactions on Information Theory}, vol.~46, no.~7, pp.~2415--2427, 2000.

\bibitem{NagaokaHayashi}
H.~Nagaoka and M.~Hayashi, ``An information-spectrum approach to classical and
  quantum hypothesis testing for simple hypotheses,'' {\em IEEE Transactions on
  Information Theory}, vol.~53, no.~2, pp.~534--549, 2007.

\bibitem{HayashiNagaoka}
M.~Hayashi and H.~Nagaoka, ``General formulas for capacity of classical-quantum
  channels,'' {\em IEEE Transactions on Information Theory}, vol.~49, no.~7,
  pp.~1753--1768, 2003.

\bibitem{WangRenner2012one}
L.~Wang and R.~Renner, ``One-shot classical-quantum capacity and hypothesis
  testing,'' {\em Physical Review Letters}, vol.~108, no.~20, p.~200501, 2012.

\bibitem{TomamichelHayashi2013hierarchy}
M.~Tomamichel and M.~Hayashi, ``A hierarchy of information quantities for
  finite block length analysis of quantum tasks,'' {\em IEEE Transactions on
  Information Theory}, vol.~59, no.~11, pp.~7693--7710, 2013.

\bibitem{DKFRR2014generalized}
F.~Dupuis, L.~Kraemer, P.~Faist, J.~M. Renes, and R.~Renner, ``Generalized
  entropies,'' in {\em XVIIth International Congress on Mathematical Physics},
  pp.~134--153, World Scientific, 2014.

\bibitem{MatthewsWehner2014finite}
W.~Matthews and S.~Wehner, ``Finite blocklength converse bounds for quantum
  channels,'' {\em IEEE Transactions on Information Theory}, vol.~60, no.~11,
  pp.~7317--7329, 2014.

\bibitem{HanVerdu}
T.~Han and S.~Verdu, ``Approximation theory of output statistics,'' {\em IEEE
  Transactions on Information Theory}, vol.~39, no.~3, pp.~752--772, 1993.

\bibitem{Hanbook}
T.~S. Han, {\em Information-Spectrum Methods in Information Theory}.
\newblock Springer, 2003.

\bibitem{HiaiPetz1991proper}
F.~Hiai and D.~Petz, ``The proper formula for relative entropy and its
  asymptotics in quantum probability,'' {\em Communications in Mathematical
  Physics}, vol.~143, no.~1, pp.~99--114, 1991.

\bibitem{OgawaNagaoka2000strong}
T.~Ogawa and H.~Nagaoka, ``Strong converse and {Stein's} lemma in quantum
  hypothesis testing,'' {\em IEEE Transactions on Information Theory}, vol.~46,
  no.~7, pp.~2428--2433, 2000.

\bibitem{Li2014second}
K.~Li, ``Second-order asymptotics for quantum hypothesis testing,'' {\em The
  Annals of Statistics}, vol.~42, no.~1, pp.~171--189, 2014.

\bibitem{NussbaumSzkola2009chernoff}
M.~Nussbaum and A.~Szko{\l}a, ``The {Chernoff} lower bound for symmetric
  quantum hypothesis testing,'' {\em The Annals of Statistics}, vol.~37, no.~2,
  pp.~1040--1057, 2009.

\bibitem{AKCMBMA2007discriminating}
K.~M. Audenaert, J.~Calsamiglia, R.~Munoz-Tapia, E.~Bagan, L.~Masanes, A.~Acin,
  and F.~Verstraete, ``Discriminating states: The quantum chernoff bound,''
  {\em Physical Review Letters}, vol.~98, no.~16, p.~160501, 2007.

\bibitem{Nagaoka2006converse}
H.~Nagaoka, ``The converse part of the theorem for quantum {Hoeffding} bound,''
  {\em arXiv preprint quant-ph/0611289}, 2006.

\bibitem{Hayashi2007error}
M.~Hayashi, ``Error exponent in asymmetric quantum hypothesis testing and its
  application to classical-quantum channel coding,'' {\em Physical Review A},
  vol.~76, no.~6, p.~062301, 2007.

\bibitem{ANSV2008asymptotic}
K.~M. Audenaert, M.~Nussbaum, A.~Szko{\l}a, and F.~Verstraete, ``Asymptotic
  error rates in quantum hypothesis testing,'' {\em Communications in
  Mathematical Physics}, vol.~279, no.~1, pp.~251--283, 2008.

\bibitem{MosonyiOgawa2015quantum}
M.~Mosonyi and T.~Ogawa, ``Quantum hypothesis testing and the operational
  interpretation of the quantum {R{\'e}nyi} relative entropies,'' {\em
  Communications in Mathematical Physics}, vol.~334, no.~3, pp.~1617--1648,
  2015.

\bibitem{Hayashi2016}
M.~Hayashi, ``Security analysis of $\varepsilon$-almost dual universal2 hash
  functions: smoothing of min entropy versus smoothing of {R\'enyi} entropy of
  order 2,'' {\em IEEE Transactions on Information Theory}, vol.~62, no.~6,
  pp.~3451--3476, 2016.

\bibitem{TSSR2011leftover}
M.~Tomamichel, C.~Schaffner, A.~Smith, and R.~Renner, ``Leftover hashing
  against quantum side information,'' {\em IEEE Transactions on Information
  Theory}, vol.~57, no.~8, pp.~5524--5535, 2011.

\bibitem{MDSFT2013on}
M.~M{\"u}ller-Lennert, F.~Dupuis, O.~Szehr, S.~Fehr, and M.~Tomamichel, ``On
  quantum {R{\'e}nyi} entropies: A new generalization and some properties,''
  {\em Journal of Mathematical Physics}, vol.~54, no.~12, p.~122203, 2013.

\bibitem{WWY2014strong}
M.~M. Wilde, A.~Winter, and D.~Yang, ``Strong converse for the classical
  capacity of entanglement-breaking and {Hadamard} channels via a sandwiched
  {R{\'e}nyi} relative entropy,'' {\em Communications in Mathematical Physics},
  vol.~331, no.~2, pp.~593--622, 2014.

\bibitem{BBCM1995generalized}
C.~H. Bennett, G.~Brassard, C.~Cr{\'e}peau, and U.~M. Maurer, ``Generalized
  privacy amplification,'' {\em IEEE Transactions on Information theory},
  vol.~41, no.~6, pp.~1915--1923, 1995.

\bibitem{DevetakWinter2005distillation}
I.~Devetak and A.~Winter, ``Distillation of secret key and entanglement from
  quantum states,'' {\em Proceedings of the Royal Society A: Mathematical,
  Physical and Engineering Sciences}, vol.~461, no.~2053, pp.~207--235, 2005.

\bibitem{Hayashi2015precise}
M.~Hayashi, ``Precise evaluation of leaked information with secure randomness
  extraction in the presence of quantum attacker,'' {\em Communications in
  Mathematical Physics}, vol.~333, no.~1, pp.~335--350, 2015.

\bibitem{Dupuis2021privacy}
F.~Dupuis, ``Privacy amplification and decoupling without smoothing,'' 
  {\em arXiv:2105.05342}, 2021.

\bibitem{Gallager1968information}
R.~Gallager, {\em Information Theory and Reliable Communication}.
\newblock John Wiley \& Sons, 1968.

\bibitem{HayashiTan2016equivocations}
M.~Hayashi and V.~Y. Tan, ``Equivocations, exponents, and second-order coding
  rates under various {R{\'e}nyi} information measures,'' {\em IEEE
  Transactions on Information Theory}, vol.~63, no.~2, pp.~975--1005, 2016.

\bibitem{MosonyiOgawa2015two}
M.~Mosonyi and T.~Ogawa, ``Two approaches to obtain the strong converse
  exponent of quantum hypothesis testing for general sequences of quantum
  states,'' {\em IEEE Transactions on Information Theory}, vol.~61, no.~12,
  pp.~6975--6994, 2015.

\bibitem{CMW2016strong}
T.~Cooney, M.~Mosonyi, and M.~M. Wilde, ``Strong converse exponents for a
  quantum channel discrimination problem and quantum-feedback-assisted
  communication,'' {\em Communications in Mathematical Physics}, vol.~344,
  no.~3, pp.~797--829, 2016.

\bibitem{HayashiTomamichel2016correlation}
M.~Hayashi and M.~Tomamichel, ``Correlation detection and an operational
  interpretation of the {R{\'e}nyi} mutual information,'' {\em Journal of
  Mathematical Physics}, vol.~57, no.~10, p.~102201, 2016.

\bibitem{MosonyiOgawa2017strong}
M.~Mosonyi and T.~Ogawa, ``Strong converse exponent for classical-quantum
  channel coding,'' {\em Communications in Mathematical Physics}, vol.~355,
  no.~1, pp.~373--426, 2017.

\bibitem{CHDH2020non}
H.-C. Cheng, E.~P. Hanson, N.~Datta, and M.-H. Hsieh, ``Non-asymptotic
  classical data compression with quantum side information,'' {\em IEEE
  Transactions on Information Theory}, vol.~67, no.~2, pp.~902--930, 2020.

\bibitem{LiYao2021reliability}
K.~Li and Y.~Yao, ``Reliability function of quantum information decoupling
  via the sandwiched {R{\'e}nyi} divergence,''
  {\em arXiv:2111.06343}, 2021.

\bibitem{GLN2005distance}
A.~Gilchrist, N.~K. Langford, and M.~A. Nielsen, ``Distance measures to compare
  real and ideal quantum processes,'' {\em Physical Review A}, vol.~71, no.~6,
  p.~062310, 2005.

\bibitem{FuchsVan1999cryptographic}
C.~A. Fuchs and J.~Van De~Graaf, ``Cryptographic distinguishability measures
  for quantum-mechanical states,'' {\em IEEE Transactions on Information
  Theory}, vol.~45, no.~4, pp.~1216--1227, 1999.

\bibitem{BCFJS1996noncommuting}
H.~Barnum, C.~M. Caves, C.~A. Fuchs, R.~Jozsa, and B.~Schumacher,
  ``Noncommuting mixed states cannot be broadcast,'' {\em Physical Review
  Letters}, vol.~76, no.~15, p.~2818, 1996.

\bibitem{Uhlmann1976the}
A.~Uhlmann, ``The `transition probability' in the state space of a
  $\ast$-algebra,'' {\em Reports on Mathematical Physics}, vol.~9, no.~2,
  pp.~273--279, 1976.

\bibitem{Hayashi2002optimal}
M.~Hayashi, ``Optimal sequence of quantum measurements in the sense of
  {Stein's} lemma in quantum hypothesis testing,'' {\em Journal of Physics A:
  Mathematical and General}, vol.~35, no.~50, p.~10759, 2002.

\bibitem{Umegaki1954conditional}
H.~Umegaki, ``Conditional expectation in an operator algebra,'' {\em Tohoku
  Mathematical Journal, Second Series}, vol.~6, no.~2-3, pp.~177--181, 1954.

\bibitem{TBH2014relating}
M.~Tomamichel, M.~Berta, and M.~Hayashi, ``Relating different quantum
  generalizations of the conditional {R{\'e}nyi} entropy,'' {\em Journal of
  Mathematical Physics}, vol.~55, no.~8, p.~082206, 2014.

\bibitem{Beigi2013sandwiched}
S.~Beigi, ``Sandwiched {R{\'e}nyi} divergence satisfies data processing
  inequality,'' {\em Journal of Mathematical Physics}, vol.~54, no.~12,
  p.~122202, 2013.

\bibitem{FrankLieb2013monotonicity}
R.~L. Frank and E.~H. Lieb, ``Monotonicity of a relative {R{\'e}nyi} entropy,''
  {\em Journal of Mathematical Physics}, vol.~54, no.~12, p.~122201, 2013.

\bibitem{LWD2016strong}
F.~Leditzky, M.~M. Wilde, and N.~Datta, ``Strong converse theorems using
  {R{\'e}nyi} entropies,'' {\em Journal of Mathematical Physics}, vol.~57,
  no.~8, p.~082202, 2016.

\bibitem{CoverThomas1991elements}
T.~M. Cover and J.~A. Thomas, {\em Elements of Information Theory}.
\newblock John Wiley \& Sons, New York, 1991.

\bibitem{DattaRenner2009smooth}
N.~Datta and R.~Renner, ``Smooth entropies and the quantum information
  spectrum,'' {\em IEEE Transactions on Information Theory}, vol.~55, no.~6,
  pp.~2807--2815, 2009.

\bibitem{ABJT2020partially}
A.~Anshu, M.~Berta, R.~Jain, and M.~Tomamichel, ``Partially smoothed
  information measures,'' {\em IEEE Transactions on Information Theory},
  vol.~66, no.~8, pp.~5022--5036, 2020.

\bibitem{Hayashi2011exponential}
M.~Hayashi, ``Exponential decreasing rate of leaked information in universal
  random privacy amplification,'' {\em IEEE Transactions on Information
  Theory}, vol.~57, no.~6, pp.~3989--4001, 2011.

\bibitem{CarterWegman1979universal}
J.~L. Carter and M.~N. Wegman, ``Universal classes of hash functions,'' {\em
  Journal of Computer and System Sciences}, vol.~18, no.~2, pp.~143--154, 1979.

\bibitem{PortmannRenner2014cryptographic}
C.~Portmann and R.~Renner, ``Cryptographic security of quantum key
  distribution,'' {\em arXiv preprint arXiv:1409.3525}, 2014.

\bibitem{Wyner1975wire}
A.~D. Wyner, ``The wire-tap channel,'' {\em Bell system technical journal},
  vol.~54, no.~8, pp.~1355--1387, 1975.

\bibitem{Hayashi2013}
M.~Hayashi, ``Tight exponential analysis of universally composable privacy
  amplification and its applications,'' {\em IEEE Transactions on Information
  Theory}, vol.~59, no.~11, pp.~7728--7746, 2013.

\bibitem{Hayashi2014}
M.~Hayashi, ``Large deviation analysis for quantum security via smoothing of
  {R\'enyi} entropy of order 2,'' {\em IEEE Transactions on Information
  Theory}, vol.~60, no.~10, pp.~6702--6732, 2014.

\end{thebibliography}
\end{document}